\documentclass{article}
\pdfoutput=1
\usepackage[english]{babel}
\usepackage[utf8]{inputenc}
\usepackage[T1]{fontenc}
\usepackage[a4paper,top=3cm,bottom=2cm,left=3cm,right=3cm,marginparwidth=1.75cm]{geometry}
\usepackage{amsmath} % Useful AMS-provided structure
\usepackage{amssymb}
\usepackage{arydshln}
\usepackage{amsthm}
\usepackage{algorithm}
\usepackage{algpseudocode}
\usepackage[
    style=alphabetic, % cite style of [ABC99]
    backref=true, % reference back to the place you cited something
    maxnames=99, % number of names before it reverts to [Aut+99]
    maxalphanames=6
]{biblatex}
\usepackage{bm} % better bolding than ams
\usepackage{booktabs} % better tables
\usepackage{braket}
\usepackage{csquotes}
\usepackage{enumitem} % helps configure lists, better than enumerate
\usepackage{etoolbox} % tex coding tools
\usepackage{eucal} % better calligraphic script
\usepackage[scaled=.92]{helvet} % helvetica
\usepackage[scaled=0.97]{inconsolata}
\usepackage{mathtools} % expands mathtools
\usepackage{mleftright}\mleftright % improves behavior of \left( and \right) (and other delimiters)
\usepackage{microtype} % fixes small typographic errors
\usepackage{parskip}
\usepackage{tikz-cd} % making arrow diagrams
\usepackage{thmtools}
\usepackage{thm-restate} % restatable environment
\usepackage{url} % defines \url{...} command for urls
\usepackage{verbatim}  % verbatim environment + comment environment
\usepackage{xcolor}
\usepackage{ninecolors}
\NineColors{saturation=high}
\definecolor{citegreen}{HTML}{208054}
\definecolor{citeblue}{HTML}{0055cc}
\usepackage{hyperref}
\usepackage[capitalize]{cleveref} % automatically resolve references to include environment name, better than prettyref
\newtheorem{theorem}{Theorem}
\newtheorem{proposition}[theorem]{Proposition}
\newtheorem{lemma}[theorem]{Lemma}
\newtheorem{fact}[theorem]{Fact}
\newtheorem{remark}[theorem]{Remark}
\newtheorem{definition}[theorem]{Definition}
\newtheorem{corollary}[theorem]{Corollary}

\hypersetup{
    breaklinks=true,   % splits links across lines
    pdfusetitle=true,  % \title and \author values into pdf metadata
    colorlinks=true, % false: boxed links; true: colored links
    linkcolor=green3, % color of internal links
    citecolor=green3, % color of links to bibliography
    urlcolor=blue3 % color of external links
}

\DefineBibliographyStrings{english}{%
  backrefpage = {page},% originally "cited on page"
  backrefpages = {pages},% originally "cited on pages"
}

\newcommand{\defeq}{\coloneqq}
\DeclarePairedDelimiter\norm{\lVert}{\rVert}

\DeclarePairedDelimiter\abs{\lvert}{\rvert}
\DeclarePairedDelimiter\parens{\lparen}{\rparen}

\DeclarePairedDelimiter\braces{\lbrace}{\rbrace}
\newcommand{\Par}[1]{\parens*{#1}}

\newcommand{\Brace}[1]{\braces*{#1}}
\newcommand{\Abs}[1]{\abs*{#1}}
\newcommand{\norms}[1]{\norm{#1}}

\newcommand{\inprod}[2]{\left\langle#1, #2\right\rangle}

\newcommand{\eps}{\varepsilon}
\newcommand{\R}{\mathbb{R}}
\newcommand{\C}{\mathbb{C}}
\newcommand{\Z}{\mathbb{Z}}
\newcommand{\xset}{\mathcal{X}}
\newcommand{\yset}{\mathcal{Y}}

\newcommand{\ma}{\mathbf{A}}
\DeclareMathOperator{\qsp}{\mathbf{QSP}}
\DeclareMathOperator{\erf}{erf}

\DeclareMathOperator{\sgn}{sgn}
\DeclareMathOperator{\poly}{poly}
\definecolor{burntorange}{rgb}{0.8, 0.33, 0.0}
\definecolor{ewin}{rgb}{0.1, 0.3, 0.65}

\newcommand{\hla}[1]{{\textcolor{olive5}{#1}}}
\newcommand{\hlb}[1]{{\textcolor{azure5}{#1}}}
\newcommand{\hlc}[1]{{\textcolor{brown5}{#1}}}
\newcommand{\vsina}{{\textcolor{olive5}{\mathcal{X}_0}}}
\newcommand{\vsinb}{{\textcolor{azure5}{\mathcal{X}_C}}}
\newcommand{\vsinc}{{\textcolor{brown5}{\mathcal{X}_1}}}
\newcommand{\vsouta}{{\textcolor{olive5}{\mathcal{Y}_0}}}
\newcommand{\vsoutb}{{\textcolor{azure5}{\mathcal{Y}_C}}}
\newcommand{\vsoutc}{{\textcolor{brown5}{\mathcal{Y}_1}}}

\newcommand{\ii}{{\imath}}
\newcommand{\mmu}{\mathbf{U}}
\newcommand{\mh}{\mathbf{H}}
\newcommand{\mpi}{\boldsymbol{\Pi}}
\newcommand{\msigma}{\boldsymbol{\sigma}}
\newcommand{\mv}{\mathbf{V}}
\newcommand{\mw}{\mathbf{W}}
\newcommand{\mx}{\mathbf{X}}
\newcommand{\my}{\mathbf{Y}}

\newcommand{\plef}{{\textup{L}}}
\newcommand{\prig}{{\textup{R}}}
\newcommand{\sv}{{\textup{(SV)}}}

\newcommand{\diag}[1]{\textbf{\textup{diag}}\left(#1\right)}
\newcommand{\half}{\frac{1}{2}}
\newcommand{\thalf}{\tfrac{1}{2}}

\newcommand{\id}{\mathbf{I}}

\newcommand{\mb}{\mathbf{B}}
\newcommand{\md}{\mathbf{D}}
\newcommand{\ms}{\mathbf{S}}
\newcommand{\mc}{\mathbf{C}}
\newcommand{\mzero}{\bm{0}}
\newcommand{\Span}{\textup{Span}}
\newcommand{\Image}{\textup{Image}}

\newcommand{\mr}{\mathbf{R}}
\newcommand{\mqu}{\bm{?}}
\newcommand{\mbl}{\mb_{\textup{L}}}
\newcommand{\mblo}{\mb_{\textup{L}, 1}}
\newcommand{\mblt}{\mb_{\textup{L}, 2}}
\newcommand{\mbr}{\mb_{\textup{R}}}
\newcommand{\mbro}{\mb_{\textup{R}, 1}}
\newcommand{\mbrt}{\mb_{\textup{R}, 2}}
\newcommand{\mproj}{\boldsymbol{\Pi}}
\newcommand{\mprojl}{\mproj_{\textup{L}}}
\newcommand{\mprojr}{\mproj_{\textup{R}}}
\newcommand{\bmb}{\overline{\mb}}
\newcommand{\bmbl}{\bmb_{\textup{L}}}
\newcommand{\bmblo}{\bmb_{\textup{L}, 1}}

\newcommand{\bmbr}{\bmb_{\textup{R}}}
\newcommand{\bmbro}{\bmb_{\textup{R}, 1}}

\newcommand{\bmproj}{\overline{\boldsymbol{\Pi}}}
\newcommand{\bmprojl}{\bmproj_{\textup{L}}}
\newcommand{\bmprojr}{\bmproj_{\textup{R}}}
\newcommand{\bmmu}{\overline{\mmu}}
\newcommand{\calS}{\mathcal{S}}
\newcommand{\dd}{\mathrm{d}}
\newcommand{\tp}{\widetilde{p}}
\newcommand{\trho}{\widetilde{\rho}}
\newcommand{\hp}{\widehat{p}}
\newcommand{\tdelta}{\widetilde{\delta}}
\newcommand{\barcsin}{\overline{\arcsin}}
\title{A CS guide to the quantum singular value transformation}
\author{Ewin Tang\thanks{University of Washington, \texttt{ewint@cs.washington.edu}. } \and Kevin Tian\thanks{Microsoft Research, \texttt{tiankevin@microsoft.com}}}
\date{}

\begin{document}

\maketitle

\begin{abstract}
    We present a simplified exposition of some pieces of \cite{gslw18}, which introduced a quantum singular value transformation (QSVT) framework for applying polynomial functions to block-encoded matrices. The QSVT framework has garnered substantial recent interest from the quantum algorithms community, as it was demonstrated by \cite{gslw18} to encapsulate many existing algorithms naturally phrased as an application of a matrix function. 
    First, we posit that the lifting of quantum singular processing (QSP) to QSVT is better viewed not through Jordan's lemma \cite{jordan75, reg06} (as was suggested by \cite{gslw18}) but as an application of the \emph{cosine-sine decomposition}, which can be thought of as a more explicit and stronger version of Jordan's lemma.
    Second, we demonstrate that the constructions of bounded polynomial approximations given in \cite{gslw18}, which use a variety of ad hoc approaches drawing from Fourier analysis, Chebyshev series, and Taylor series, can be unified under the framework of \emph{truncation of Chebyshev series}, and indeed, can in large part be matched via a bounded variant of a standard meta-theorem from \cite{trefethen19}. We hope this work finds use to the community as a companion guide for understanding and applying the powerful framework of \cite{gslw18}.
\end{abstract}

% plaintext version
% We present a simplified exposition of some pieces of [Gilyén, Su, Low, and Wiebe, STOC'19, arXiv:1806.01838], who introduced a quantum singular value transformation (QSVT) framework for applying polynomial functions to block-encoded matrices. The QSVT framework has garnered substantial recent interest from the quantum algorithms community, as it was demonstrated by [gslw18] to encapsulate many existing algorithms naturally phrased as an application of a matrix function. First, we posit that the lifting of quantum singular processing (QSP) to QSVT is better viewed not through Jordan's lemma (as was suggested by [gslw18]) but as an application of the cosine-sine decomposition, which can be thought of as a more explicit and stronger version of Jordan's lemma. Second, we demonstrate that the constructions of bounded polynomial approximations given in [gslw18], which use a variety of ad hoc approaches drawing from Fourier analysis, Chebyshev series, and Taylor series, can be unified under the framework of truncation of Chebyshev series, and indeed, can in large part be matched via a bounded variant of a standard meta-theorem from [Trefethen, 2013]. We hope this work finds use to the community as a companion guide for understanding and applying the powerful framework of [gslw18].

\tableofcontents
\newpage

\section{Introduction}\label{sec:intro}

We present a ``user-friendly guide'' to understanding technical aspects of the \emph{quantum singular value transformation} (QSVT), an elegant framework for designing quantum algorithms, particularly those that can be phrased as a \emph{linear algebraic task} on a quantum state $\ket{\psi}$, viewed as a vector of amplitudes $\sum_{k=1}^d \psi_k \ket{k}$.
This includes Hamiltonian simulation~\cite{LowC17}, i.e.\ preparing $e^{\ii t\mh}\ket{\psi}$ for a Hamiltonian $\mh$; quantum linear system solving~\cite{hhl09}, i.e.\ preparing $\ma^{-1}\ket{\psi}$ for a sparse matrix $\ma$; and quantum random walks \cite{Szegedy04}, i.e.\ approximating large powers of a Markov chain transition matrix or discriminating its singular values.
QSVT was introduced by \cite{lc16} in the Hermitian case and then generalized and subsequently popularized by a paper of Gily\'{e}n, Su, Low, and Wiebe~\cite{gslw18} which demonstrates that these important, seemingly disparate quantum algorithms can be seen as consequences of a single unifying primitive. 

Our aim is to expose the beauty of~\cite{gslw18} and make it more accessible to an audience with a background in computer science, by simplifying or providing alternatives to its more mathematically dense proofs. This goal may be viewed as complementary to prior expositions of \cite{gslw18} such as \cite{MartynRTC21}, which focused on describing applications of QSVT to the design of quantum algorithms. In contrast, the aim of our work is to directly provide streamlined proofs or alternatives to the main technical results in \cite{gslw18}. 

\subsection{Our contributions}

\paragraph{QSVT via the CS decomposition.} In Section~\ref{sec:csd}, we give an alternate exposition of the \emph{qubitization} technique given in \cite[Section 3.2]{gslw18}. This technique lifts \emph{quantum signal processing}, a product decomposition for computing bounded scalar polynomials, to QSVT, its matrix counterpart. In particular, QSVT implements quantum signal processing separately on each of the singular values of a ``block encoded matrix'' by mapping them through a polynomial transformation, while preserving the block encoding structure. Our exposition of QSVT is by way of the \emph{Cosine-Sine} decomposition, a strengthening of Jordan's lemma \cite{jordan75} that is more amenable to the computations in the proof of QSVT's correctness. We believe our proof strategy simplifies the exposition of QSVT in Section 3.2 of \cite{gslw18} (and other related expositions, e.g.\ Chapters 7 and 8 of \cite{Lin22}). Specifically, the viewpoint we adopt elucidates the action of QSVT on the block structure of encoded matrices, removing much of the casework and eigenspace-by-eigenspace reasoning of prior expositions.

\paragraph{Bounded approximations via truncated Chebyshev series.} In Section~\ref{sec:cheb}, we apply the technique of truncating \emph{Chebyshev Series} to match or nearly-match all of the polynomial approximation results needed throughout \cite[Section 5]{gslw18} through an arguably simpler framework.
Our starting point is a classical theorem of Trefethen (\cref{thm:trefethen}) which bounds the error incurred by Chebyshev truncation for smooth functions. We derive, as a consequence of Trefethen's result, a new ``bounded Chebyshev truncation'' analog (\cref{thm:main_bounded}) which applies to piecewise-smooth functions, and is compatible with the QSVT framework. Unlike its analog in the original work \cite[Corollary 66]{gslw18}, our \cref{thm:main_bounded} does not use Taylor series or Fourier series in its proof: the only approximation theory tools used are standard properties of Chebyshev polynomials. Our result is comparable to \cite[Corollary 66]{gslw18}; as discussed in Remark~\ref{rem:bounded_comments}, it loses a logarithmic factor in some regimes, but uses a weaker assumption on the function to be approximated. We further provide a user's guide on how to apply \cref{thm:main_bounded} to derive bounded approximations to functions of interest in applications.

In Section~\ref{ssec:separation}, we give a new separation result lower bounding the degree of polynomials approximating the exponential function $\exp$, under a boundedness requirement. Boundedness is crucial in QSVT applications (see Remark~\ref{rmk:qsvt-to-alg}) due to the spectra of quantum gates. Notably, bounded approximations to $\exp$ have found use in designing Gibbs sampling techniques for quantum optimization \cite{vanApeldoornG19, BoulandGJSWT23}. We show that in parameter regimes of interest for these applications, a quadratically larger degree is required to achieve a bounded approximating polynomial. While lower bounds on QSVT have previously been demonstrated \cite[Section 6]{gslw18}), our separation result is purely an approximation theory statement (independent of its use in QSVT), and its (simple) proof uses different techniques than \cite{gslw18}. We hope this result sheds light on when one can hope to obtain low-degree approximations for QSVT.

%This may be unsurprising from conventional wisdom in approximation theory, as Chebyshev, Fourier, and Laurent/Taylor series are ``essentially equivalent'' under changes of variables, as discussed in \cite[Preface]{trefethen19}.

%\ewin{Include below in the submission?}
\iffalse
We recommend viewing this note as a companion piece to \cite{gslw18}, for readers interested in understanding or applying QSVT.
Correspondingly, our exposition will be lighter on discussing motivations than is appropriate for a self-contained work, as we focus our scope on only what is needed to provide context for the techniques of \cite{gslw18}. Though this choice does limit the self-contained readability of this work, we hope it clarifies the ideas underlying QSVT by cleaning up technical details and providing potentially simpler, user-friendly alternatives to prior developments.\footnote{As our goal is expository in nature, we would like to encourage the reader to contact us with feedback on any portions of our work which could benefit from further clarification.}
\fi
\subsection{Notation}

Matrices have bolded variable names, and $\id$ is the identity matrix with dimension specified by context.
For a matrix $\mmu$, $\mmu^\dagger$ denotes its conjugate transpose. 
A square matrix $\mmu$ is \emph{unitary} if $\mmu^\dagger \mmu = \mmu \mmu^\dagger = \id$, and we call it ``partitioned'' if it has a block matrix structure with two row blocks and two column blocks (which is clear from context).
When writing a matrix as blocks, an empty block denotes a zero block, and $\cdot$ denotes that the block contains arbitrary entries.
For brevity, we omit the dimensions of blocks when these sizes are not important to the computation: all block matrices occurring in products are compatible in the standard way.

The ``computational basis'' is the standard basis in $\C^d$.
We denote $[d] \defeq \{1,2,\ldots,d\}$, $\ii \defeq \sqrt{-1}$, and $\msigma_z$ is the Pauli matrix $(\begin{smallmatrix}
1 & 0 \\ 0 & -1
\end{smallmatrix})$. 
The maximum absolute value of a real function $f$ on $[a, b]$ is denoted $\norm{f}_{[a, b]}$. We write $a \eqsim b$ to mean there are universal constants $0 < C_1 \le C_2$ with $C_1 a \le b \le C_2 a$.
\section{QSVT and the Cosine-Sine decomposition}\label{sec:csd}

\begin{quotation}
    \textit{Briefly, whenever some aspect of a problem can be usefully formulated in terms of two-block by two-block partitions of unitary matrices, the CS decomposition will probably add insights and simplify the analysis.}
    \hspace{1em plus 1fill}---Paige and Wei, \cite{pw94}
\end{quotation}

\subsection{Existence of the CS decomposition} \label{ssec:csd}

We begin by proving the existence of the CS decomposition (CSD), a decomposition of a partitioned unitary matrix, following Paige and Wei~\cite{pw94}. We describe its application to the quantum singular value transformation (QSVT) in the following Sections~\ref{sec:qsvt},~\ref{ssec:basic_case}, and~\ref{ssec:gen_case}, wherein we give an alternate proof of the main QSVT result in \cite{gslw18} stated as Theorem~\ref{thm:qsp-to-qsvt}.

The main idea of the CSD is that when a unitary matrix $\mmu$ is split into two-by-two blocks $\mmu_{ij}$ for $i, j \in \{1, 2\}$, one can produce ``simultaneous singular value decompositions (SVDs)'' of the blocks, of the form $\mmu_{ij} = \mv_i \md_{ij} \mw_j^\dagger$.\footnote{In fact, there is some sense in which the SVD and the CSD are special cases of the same object, a \emph{generalized Cartan decomposition}. We recommend the survey by Edelman and Jeong for readers curious about this connection~\cite{ej21}.}
If the reader cares just about the application to QSVT, they can read Theorem~\ref{thm:cs} and skip to \cref{sec:qsvt}. 

For additional intuition on the CSD, we refer the reader to Appendix~\ref{app:csd_interpret}, in which we derive principal angles and Jordan's lemma as consequences.

\begin{theorem}\label{thm:cs}
Let $\mmu \in \C^{d \times d}$ be a unitary matrix, partitioned into blocks of size $\{r_1, r_2\} \times \{c_1, c_2\}$:
\[\mmu = \begin{pmatrix} \mmu_{11} & \mmu_{12} \\ \mmu_{21} & \mmu_{22}\end{pmatrix},\text{ where } \mmu_{ij} \in \C^{r_i \times c_j} \text{ for } i,j\in\{1,2\}.\]
Then, there exists unitary $\mv_i \in \C^{r_i \times r_i}$ and $\mw_j \in \C^{c_j \times c_j}$ for $i, j \in \{1,2\}$ such that
\[
    \begin{pmatrix} \mmu_{11} & \mmu_{12} \\ \mmu_{21} & \mmu_{22} \end{pmatrix}
    = \begin{pmatrix} \mv_1 & \\ & \mv_2 \end{pmatrix}
    \begin{pmatrix} \md_{11} & \md_{12} \\ \md_{21} & \md_{22} \end{pmatrix}
    \begin{pmatrix} \mw_1 & \\ & \mw_2 \end{pmatrix}^\dagger,
\]
where blanks represent zero matrices and $\md_{ij} \in \R^{r_i \times c_j}$ are diagonal matrices, possibly padded with zero rows or columns.
Specifically, we can write 
\begin{equation} \label{eq:d-form}
    \md \defeq \begin{pmatrix} \md_{11} & \md_{12} \\ \md_{21} & \md_{22} \end{pmatrix}
    =
    \left(\begin{array}{@{}c|c@{}}
    \begin{matrix}
    \mzero & & \\
    & \mc & \\
    & & \id
    \end{matrix} &
    \begin{matrix}
    \id & & \\
    & \hphantom{-}\ms & \\
    & & \hphantom{-}\mzero
    \end{matrix} \\
    \hline
    \begin{matrix}
    \id & & \\
    & \ms & \\
    & & \mzero
    \end{matrix} &
    \begin{matrix}
    \mzero & & \\
    & -\mc & \\
    & & -\id
    \end{matrix}
    \end{array}
    \right)
\end{equation}
where $\id$, $\mc$, and $\ms$ blocks are square diagonal matrices where $\mc$ and $\ms$ have entries in $(0,1)$ on the diagonal, and $\mzero$ blocks may be rectangular.\footnote{Blocks may be non-existent. The $\id$ blocks may not necessarily be the same size, but $\mc$ and $\ms$ are the same size.}
Because $\md$ is unitary, we also have $\mc^2 + \ms^2 = \id$.
\end{theorem}

\begin{remark} \label{rmk:cs-rotation}
    The form of $\md$ naturally induces decompositions $\mathbb{C}^d = \vsina \oplus \vsinb \oplus \vsinc$ and $\mathbb{C}^d = \vsouta \oplus \vsoutb \oplus \vsoutc$ into direct sums of three spaces.
    Hence, $\md: \C^d \to \C^d$ can be seen as a map $\md: \vsina \oplus \vsinb \oplus \vsinc \to \vsouta \oplus \vsoutb \oplus \vsoutc$, such that $\md$ is a direct sum of three linear maps.
\[
    \left(\begin{array}{@{}c|c@{}}
    \begin{matrix}
    \hla{\mzero} & & \\
    & \hlb{\mc} & \\
    & & \hlc{\id}
    \end{matrix} &
    \begin{matrix}
    \hla{\id} & & \\
    & \hphantom{-}\hlb{\ms} & \\
    & & \hphantom{-}\hlc{\mzero}
    \end{matrix} \\
    \hline
    \begin{matrix}
    \hla{\id} & & \\
    & \hlb{\ms} & \\
    & & \hlc{\mzero}
    \end{matrix} &
    \begin{matrix}
    \hla{\mzero} & & \\
    & \hlb{-\mc} & \\
    & & \hlc{-\id}
    \end{matrix}
    \end{array}
    \right)
    = \underbrace{\hla{\begin{pmatrix}
        \mzero & \id \\ \id & \mzero
    \end{pmatrix}}}_{\vsina \to \vsouta}
    \oplus \underbrace{\hlb{\begin{pmatrix}
        \mc & \hphantom{-}\ms \\ \ms & -\mc
    \end{pmatrix}}}_{\vsinb \to \vsoutb}
    \oplus \underbrace{\hlc{\begin{pmatrix}
        \id & \hphantom{-}\mzero \\ \mzero & -\id
    \end{pmatrix}}}_{\vsinc \to \vsoutc}.
\]
    The key resulting intuition for QSVT is that, supposing everything is square, these blocks can be further decomposed into $2\times 2$ blocks of the following rotation matrix form \[\begin{pmatrix} \lambda_i & \sqrt{1-\lambda_i^2} \\ \sqrt{1-\lambda_i^2} & -\lambda_i\end{pmatrix}\]
    from this representation, where $\{\lambda_i\}$ are the singular values of $\mmu_{11}$ (see Lemma~\ref{lem:qsp-to-qsvt}).
\end{remark}

For completeness we now recall how to derive the CS decomposition from other common decompositions (namely, the SVD and QR decompositions), following the proof of \cite{pw94}.\footnote{Our exposition will be slightly different, since Paige and Wei use a QR decomposition with a \emph{lower} triangular matrix, where we use the more standard one with an \emph{upper} triangular matrix.}
\begin{proof}[Proof of Theorem~\ref{thm:cs}]
We begin by considering $\mmu_{11} \in \C^{r_1 \times c_1}$.
Let $\mv_1 \md_{11} \mw_1^\dagger$ be a SVD of $\mmu_{11}$, where $\md_{11} \in \R^{r_1 \times c_1}$ and $\mv_1$ and $\mw_1$ are square unitaries.
Since $\norm{\mmu_{11}}_{\textup{op}} \le \norm{\mmu}_{\textup{op}} = 1$, all its singular values are between zero and one, so we can specify that $\md_{11}$ takes the form
\[
    \md_{11} = \begin{pmatrix}
    \mzero & & \\
    & \mc & \\
    & & \id
    \end{pmatrix}
\]
for a diagonal matrix $\mc$ with $0 < \mc_{ii} < 1$ for all $i$.
Now, take QR decompositions of $\mmu_{21}\mw_1 \in \C^{r_2 \times c_1}$ and $\mmu_{12}^\dagger \mv_1 \in \C^{c_2 \times r_1}$.
These decompositions give unitaries $\mv_2$ and $\mw_2$ such that $\md_{21} \defeq \mv_2^\dagger \mmu_{21} \mw_1$ and $\md_{12}^\dagger \defeq \mw_2^\dagger \mmu_{12}^\dagger \mv_1$ are upper triangular with non-negative entries on the diagonal.
By design, the $\mv_{i}$ and $\mw_j$ we have defined give us the decomposition
\begin{align*}
    \begin{pmatrix} \mv_1 \\ & \mv_2\end{pmatrix}^\dagger
    \begin{pmatrix} \mmu_{11} & \mmu_{12} \\ \mmu_{21} & \mmu_{22} \end{pmatrix}
    \begin{pmatrix} \mw_1 \\ & \mw_2 \end{pmatrix}
    = \underbrace{\begin{pmatrix} \md_{11} & \md_{12} \\ \md_{21} & \mv_2^\dagger \mmu_{22}\mw_2 \end{pmatrix}}_{\md}.
\end{align*}
We claim that the blocks $\md_{21}$ and $\md_{12}$ satisfy the desired form from \eqref{eq:d-form}.
This will (almost) be our final decomposition.
We will only give the argument for $\md_{21}$; the one for $\md_{12}$ is similar.
The columns of $\md$ are orthonormal and $\md_{21}$ is upper triangular with non-negative entries on its diagonal.
So, all of the rightmost blocks of $\md_{21}$ (under the $\id$ from $\md_{11}$) must be zero because of orthonormality, and further, the top-left block of $\md_{21}$ must be $\id$, using upper triangularity and non-negativity of the diagonal (inducting by row, the top-left entry is $1$, forcing its row to be zeroes, and so on down the diagonal).
Since the rows of $\md$ are orthonormal, the $\id$ block in $\md_{21}$ forces the block to its right to be the zero matrix.
Finally, upper triangularity and column orthonormality forces the middle block to be non-negative diagonal with $\ms^2 + \mc^2 = \id$.
The logic is displayed below:
\[\md_{21} =
    \begin{pmatrix}
    \mqu & \mqu & \mqu \\
    & \mqu & \mqu\\
    & & \mqu
    \end{pmatrix} \to
    \begin{pmatrix}
    \mqu & \mqu & \\
    & \mqu & \\
    & & \mzero
    \end{pmatrix} \to
    \begin{pmatrix}
    \id & \mqu & \\
    & \mqu & \\
    & & \mzero
    \end{pmatrix} \to
    \begin{pmatrix}
    \id & & \\
    & \mqu & \\
    & & \mzero
    \end{pmatrix}
    \to
    \begin{pmatrix}
    \id & & \\
    & \ms & \\
    & & \mzero
    \end{pmatrix}.\]
What we have argued so far suffices to show that (recalling $\md$ is unitary)
\begin{equation*}
    \md =
    \left(\begin{array}{@{}c|c@{}}
    \begin{matrix}
    \mzero & & \\
    & \mc & \\
    & & \id
    \end{matrix} &
    \begin{matrix}
    \id & & \\
    & \ms & \hphantom{\mqu_{12}} \\
    & \hphantom{\mqu_{21}} & \mzero
    \end{matrix} \\
    \hline
    \begin{matrix}
    \id & & \\
    & \ms & \\
    & & \mzero
    \end{matrix} &
    \begin{matrix}
    \mzero & & \\
    & \mqu_{11} & \mqu_{12} \\
    & \mqu_{21} & \mqu_{22}
    \end{matrix}
    \end{array}
    \right).
\end{equation*}
For brevity we will only sketch the rest of the argument about the bottom-right block $\md_{22}$. First, $\mqu_{11} = -\mc$ follows from unitarity of the following block of $\md$, which shows $\mc \ms + \ms \mqu_{11} = 0$:
\[\begin{pmatrix} \mc & \ms \\ \ms & \mqu_{11} \end{pmatrix}. \]
The blocks $\mqu_{21}$ and $\mqu_{12}$ must then be zero using unitarity, considering the fifth (block) row and column in $\md$. Finally, $\mqu_{22}$ is unitary and can be rotated to the (negative) identity by changing $\mw_2$:
\[ \mw_2 \leftarrow \begin{pmatrix}
    \id \\ & \id \\ & & -\mqu_{22}^\dagger
\end{pmatrix}\mw_2 . \qedhere \]
\end{proof}

In Appendix~\ref{app:csd_interpret}, we demonstrate that for specific choices of $\mmu$, the form of the CS decomposition reveals interesting properties about the interactions between subspaces. In particular, the diagonals of $\mc$ and $\ms$ can naturally be seen as the cosines and sines of ``principal angles'' between subspaces. For those interested in applications to physics, these cosines and sines correspond to reflection and transmission probabilities in scattering theory, where this decomposition is known as the \emph{polar decomposition}~\cite{mpk88,ml92,beenakker97}.

\subsection{QSVT and quantum signal processing}\label{sec:qsvt}

We now apply the machinery of Section~\ref{ssec:csd} to prove correctness of the QSVT framework of \cite{gslw18}.
We first recall the situation treated by QSVT, requiring the notion of a block encoding.

\begin{definition}[Variant of {\cite[Definition 43]{gslw18}}]\label{def:block_encode}
    Given $\ma \in \C^{r \times c}$, we say unitary $\mmu \in \C^{d \times d}$ is a block encoding of $\ma$ if there are $\mb_{\textup{L},1} \in \C^{d \times r}, \mb_{\textup{R},1} \in \C^{d \times c}$ with orthonormal columns such that $\mb_{\textup{L},1}^\dagger \mmu \mb_{\textup{R},1} = \ma$. We denote $\mprojl = \mblo\mblo^\dagger$, $\mprojr = \mbro\mbro^\dagger$ to be the corresponding projections onto the spans of $\mblo$ and $\mbro$, respectively.
\end{definition}

In other words, if $\mmu$ is a block encoding of $\ma$, then in the right basis, it has $\ma$ as a submatrix.\footnote{
    Methods for preparing block encodings often produce block encodings of $\tfrac 1 \alpha \ma$ for some scaling factor $\alpha$.
    This factor $\alpha$ appears in gate complexities of applications.
    Some authors parametrize this notion, e.g.\ by saying $\mmu$ is a $(\alpha, \eps)$-block encoding if $\mmu$ is a block encoding of some $\tilde{\ma}$ such that $\norms{\tilde{\ma} - \tfrac 1 \alpha \ma} \leq \tfrac \eps \alpha$ in operator norm.
}
That is, if $\mbl = \begin{pmatrix} \mblo & \mblt \end{pmatrix}$ and $\mbr = \begin{pmatrix} \mbro & \mbrt \end{pmatrix}$ are unitary completions of $\mblo$ and $\mbrt$,
\begin{align*}
    \mbl^\dagger \mmu \mbr = \begin{pmatrix}
        \ma & \cdot \\ \cdot & \cdot
    \end{pmatrix} \text{ and }
    \mbl^\dagger (\mprojl \mmu \mprojr) \mbr = \begin{pmatrix}
        \ma & \mzero \\ \mzero & \mzero
    \end{pmatrix}.
\end{align*}
In Section~\ref{ssec:basic_case}, we consider the special case when $\mblo$ and $\mbro$ are the first $r$ and $c$ columns of the identity, respectively.
Under this restriction, the following statements about submatrices are true in the computational basis:
\begin{equation}\label{eq:compbasis}
    \mmu = \begin{pmatrix}
        \ma & \cdot \\ \cdot & \cdot
    \end{pmatrix} \text{ and }
    \mprojl \mmu \mprojr = \begin{pmatrix}
        \ma & \mzero \\ \mzero & \mzero
    \end{pmatrix}.
\end{equation}
This simplification is for the purposes of exposition, since then $\mmu$ is clearly a block matrix which we can apply the CS decomposition to. 
Indeed, it is without loss of generality: in Section~\ref{ssec:gen_case}, we recover general statements by unitary transformations which reduce to the special case above.

We now describe the QSVT framework, a lifting of a $2\times 2$ matrix polynomial construction defined via ``phase factors'' (Definition~\ref{def:qsp}), to higher dimensions. The construction in the $2 \times 2$ case is referred to in the literature as quantum signal processing (QSP). We recall the basics of QSP here.
\begin{definition}[Quantum signal processing] \label{def:qsp}
    A sequence of phase factors $\Phi = \{\phi_j\}_{0 \leq j \leq n} \in \R^{n+1}$ defines a \emph{quantum signal processing} circuit%
    \footnote{
        We define QSP with the reflection operation $\mr(x)$; a different convention is to use the rotation $e^{\ii\arccos(x)\msigma_x} = (\begin{smallmatrix} x & \ii\sqrt{1-x^2} \\ \ii\sqrt{1-x^2} & x\end{smallmatrix})$, denoted $W(x)$ in \cite{gslw18}.
        These two types of circuits are equivalent up to a shift in phase factors~\cite[Appendix A.2]{MartynRTC21}.
        Using $W(x)$ is perhaps more natural, since then this corresponds to alternating rotations in the $\sigma_X$ and $\sigma_Z$ basis.
    }
    \begin{align} \label{eq:qsp}
        \qsp(\Phi, x)
        \defeq \begin{pmatrix}
            e^{\ii \phi_0} & 0 \\ 0 & e^{-\ii \phi_0}
        \end{pmatrix}
        \prod_{j=1}^n \underbrace{\begin{pmatrix}
            x & \sqrt{1-x^2} \\ \sqrt{1-x^2} & -x
        \end{pmatrix}}_{=: \mr(x)}
        \underbrace{\begin{pmatrix}
            e^{\ii \phi_j} & 0 \\ 0 & e^{-\ii \phi_j}
        \end{pmatrix}}_{e^{\ii \phi_j \msigma_z}}.
    \end{align}
    Here and elsewhere, the product goes from $1$ on the left-hand side to $n$ on the right-hand side (by this convention, rotations are applied from $\phi_n$ to $\phi_0$).
\end{definition}

The idea of QSP is that we can perform a \emph{known} polynomial, satisfying an achievability condition defined below, on an \emph{unknown} (parameterized) operator via interleaved rotations.

\begin{definition}[QSP-achievable polynomial] \label{def:qsp-achievable}
    We say that a polynomial $p(x) \in \C[x]$ is \emph{QSP-achievable} if there is a sequence of phase factors $\Phi = \{\phi_j\}_{0 \leq j \leq n} \in \R^{n+1}$ such that
    \begin{align}
        \qsp(\Phi, x) = \begin{pmatrix}
            p(x) & \cdot \\ \cdot & \cdot
        \end{pmatrix}.
    \end{align}
\end{definition}

Through elementary calculations, \cite{gslw18} gives a characterization of QSP-achievability; in particular, they show that every bounded, real polynomial $p_{\Re}$ is achievable in the sense that there is a QSP-achievable polynomial whose real part is $p_{\Re}$, which we summarize here. For self-containedness, we give a proof in \cref{sec:qsp-proofs}, with a formal statement in \cref{thm:real-characterization}.
\begin{theorem}[{\cite[Corollary 10]{gslw18}}] \label{thm:real-one}
    Let $p_{\Re}(x) \in \mathbb{R}[x]$ be a real polynomial of degree $n$.
    Then there exists a $p \in \C[x]$ such that $p_{\Re} = \Re(p)$ and $p$ is QSP-achievable with some $\Phi = \{\phi_j\}_{0 \leq j \leq n} \in \R^{n+1}$ if and only if the following conditions hold:
    \begin{enumerate}[label=(\alph*)]
        \item $p_{\Re}$ is even or odd;
        \item $\abs{p_{\Re}(x)} \leq 1$ for $x \in [-1,1]$.
    \end{enumerate}
\end{theorem}

If $p$ is QSP-achievable, then we can achieve $\Re(p)$ via a linear combination of unitary circuits on top of QSP; see \cref{rmk:qsvt-to-alg}. We next introduce a generalization of Definition~\ref{def:qsp} to higher dimensions.
\begin{definition}[{\cite[Definition 15]{gslw18}}] \label{def:15}
The \emph{phased alternating sequence} associated with a partitioned unitary $\mmu$ (following notation of \cref{def:12}) and $\Phi = \{\phi_j\}_{0 \leq j \leq n} \in \R^{n+1}$ is
    \begin{align*}
        \mmu_\Phi &\defeq \begin{cases}
            e^{\ii\phi_0(2\mpi_\plef-\id)} \mmu e^{\ii\phi_1(2\mpi_\prig-\id)} \displaystyle\prod_{j=1}^{\frac{n - 1}{2}} \mmu^\dagger e^{\ii\phi_{2j}(2\mpi_\plef - \id)} \mmu e^{\ii\phi_{2j+1}(2\mpi_\prig - \id)} &\text{if } n \text{ is odd, and} \\
            \hspace{5.8em} e^{\ii\phi_0(2\mpi_\plef-\id)} \displaystyle\prod_{j=1}^{\frac n 2} \mmu^\dagger e^{\ii\phi_{2j-1}(2\mpi_\plef - \id)}\mmu e^{\ii\phi_{2j}(2\mpi_\prig - \id)}  &\text{if } n \text{ is even.} \\
        \end{cases}
    \end{align*}
\end{definition}

\begin{remark}\label{rem:high-d}
The phased alternating sequence $\mmu_\Phi$ can be seen as a generalization of the quantum signal processing circuit $\qsp(\Phi, x)$. When $d = 2$ and $r = c = 1$, $2\mpi_\plef - \id = 2\mpi_\prig - \id = \msigma_z$, so 
    \begin{align*}
        \qsp(\Phi,x) = [\mr(x)]_\Phi \text{ where } \mr(x) = \begin{pmatrix}
            x & \sqrt{1-x^2} \\ \sqrt{1-x^2} & -x
        \end{pmatrix}.
    \end{align*}
\end{remark}

Finally, we define the matrix polynomial we wish to target via the QSVT framework as follows.

\begin{definition}[{\cite[Definition 16]{gslw18}}]\label{def:16}
    Let $f: \R \to \C$ be even or odd, and let $\ma \in \C^{r \times c}$ have SVD $\ma = \sum_{i\in[\min(r, c)]} \sigma_iu_iv_i^\dagger$.
    Then we define
    \begin{align*}
        f^\sv(\ma) = \begin{cases}
            \sum_{i\in[\min(r, c)]} f(\sigma_i)u_iv_i^\dagger & f\text{ is odd} \\
            \sum_{i\in[c]} f(\sigma_i)v_iv_i^\dagger & f\text{ is even}
        \end{cases}
    \end{align*}
    where $\sigma_i$ is defined to be zero for $i > \min(r, c)$.
\end{definition}

When $f(x) = p(x)$ is an even or odd polynomial, $p^\sv(\ma)$ can be written as a polynomial in the expected way,
e.g.\ if $p(x) = x^2 + 1$, $p^\sv(\ma) = \ma^\dagger \ma + \id$ and if $p(x) = x^3 + x$, $p^\sv(\ma) = \ma \ma^\dagger \ma + \ma$. 

With this definition in hand, we are now ready to state the main result of the QSVT framework.
\begin{theorem}[{\cite[Theorem 17]{gslw18}}] \label{thm:qsp-to-qsvt}
    Let partitioned unitary $\mmu \in \C^{d \times d}$ be a block encoding of $\ma$.
    Let $\Phi = \{\phi_j\}_{0 \leq j \leq n} \in \R^{n+1}$ be the sequence of phase factors such that $\qsp(\Phi, x)$ computes the degree-$n$ polynomial $p(x) \in \C[x]$, as in \cref{def:qsp-achievable}.
    Then $\mmu_\Phi$ is a block encoding of $p^\sv(\ma)$:
    \begin{align*}
        \text{if }p\text{ is odd,}\quad
        \mpi_\plef \mmu_\Phi \mpi_\prig
        &= \begin{pmatrix}
            p^\sv(\ma) & \mzero \\ \mzero & \mzero
        \end{pmatrix}
        = p^\sv(\mpi_\plef \mmu \mpi_\prig), \\
        \text{and if }p\text{ is even,}\quad
        \mpi_\prig \mmu_\Phi \mpi_\prig
        &= \begin{pmatrix}
            p^\sv(\ma) & \mzero \\ \mzero & \mzero
        \end{pmatrix}
        = \mpi_\prig p^\sv(\mpi_\plef \mmu \mpi_\prig) \mpi_\prig.
    \end{align*}
\end{theorem} 
As we will see, when $\mblo$ and $\mbro$ are in the computational basis, the CS decomposition (Theorem~\ref{thm:cs}) readily reduces the proof of Theorem~\ref{thm:qsp-to-qsvt} to substantially simpler subproblems (see Lemma~\ref{lem:qsp-to-qsvt}).

\begin{remark}[QSVT to quantum algorithms] \label{rmk:qsvt-to-alg}
    \cref{thm:qsp-to-qsvt} typically admits quantum algorithms in the following way.
    Suppose our goal is to apply $\ma^{-1}$ to a quantum state $\ket{\psi}$, where $\ma$ is a matrix with singular values in $[\tfrac 1 \kappa, 1]$ that we have in the block encoding $\mmu$.
    First, we take an odd, bounded polynomial $p(x) \in \R[x]$ such that $\abs{p(x) - \tfrac \kappa x} \leq \eps$ for $x \in [\tfrac 1 \kappa, 1]$.
    Then by \cref{thm:real-one}, there is a phase sequence $\Phi$ which implements a $q \in \C[x]$ such that $\Re(q) = p$.
    By \cref{thm:qsp-to-qsvt}, $\mmu_\Phi$ is a block encoding of $q^\sv(\ma)$ and a calculation shows that $\mmu_{-\Phi}$ is a block encoding of the same polynomial but with coefficients conjugated, $[q^*]^\sv(\ma)$.
    
    The circuit $(\mh \otimes \id) (|0\rangle\langle0| \otimes \mmu_\Phi + |1\rangle\langle1| \otimes \mmu_{-\Phi}) (\mh \otimes \id)$ is a block encoding of $\half(q^\sv(\ma) + [q^*]^\sv(\ma)) = p^\sv(\ma)$, and Corollary 19 of \cite{gslw18} shows that one can implement this circuit with controlled $\mmu$ and $\mmu^\dagger$'s, along with other gates based on $\mpi_\plef$ and $\mpi_\prig$.

    Equipped with a block encoding of $p^\sv(\ma)$, we can accomplish our desired algorithmic task.
    To apply $p^\sv(\ma)$ to an input state $\ket{\psi} \in \C^c$, we rotate $\ket{\psi}$ in $d$-dimensional state space to be aligned with the block in the block encoding.
    After applying $\mmu$ and post-selecting, the output state is $p^\sv(\ma) \ket{\psi}$ normalized to have unit norm, and $\norm{p^\sv(\ma) - \ma^{-1}} \leq \eps$ due to our choice of polynomial.\footnote{A warning: when applying QSVT on an approximate block encoding of $\tilde{\ma}$ with $\norm{\tilde{\ma} - \ma} \leq \eps$, the error may not propagate as expected. This is because, if $f$ satisfies $|f'| \le L$, $\norm{f^\sv(\ma) - f^\sv(\tilde{\ma})} \leq L\norm{\ma - \tilde{\ma}}$ is \emph{not true in general}, even up to constants. As Section~3.3 of \cite{gslw18} discusses, sometimes one must lose logarithmic factors here.}
\end{remark}

\subsection{Simplified QSVT in the computational basis}\label{ssec:basic_case}

In this section, we provide a proof of Theorem~\ref{thm:qsp-to-qsvt} in the computational basis. We begin with some helpful notation in this special case, following the partitioning given by Theorem~\ref{thm:cs}.

\begin{definition}[{Variant of \cite[Definition 12]{gslw18}}]\label{def:12}
Let $\mmu \in \C^{d \times d}$ be a block encoding of $\ma \in \C^{r \times c}$ where $\mblo$ and $\mbro$ are the first $r$ and $c$ columns of the identity, respectively (see \eqref{eq:compbasis}).
By \cref{thm:cs}, there is a CS decomposition compatible with the partitioning of $\mmu$:
\[\mmu = \begin{pmatrix} \ma & \mmu_{12} \\ \mmu_{21} & \mmu_{22} \end{pmatrix} 
  = \underbrace{\begin{pmatrix} \mv_{1} \\ & \mv_{2} \end{pmatrix}}_{\mv}
    \underbrace{\begin{pmatrix} \md_{11} & \md_{12} \\ \md_{21} & \md_{22} \end{pmatrix}}_{\md}
    \underbrace{\begin{pmatrix} \mw_{1} \\ & \mw_{2} \end{pmatrix}^\dagger}_{\mw^\dagger}.\]
\end{definition}

In Definition~\ref{def:12}, we applied Theorem~\ref{thm:cs} to obtain an SVD of $\ma = \mv_1 \md_{11} \mw_1$ that we have extended to the $d$-dimensional $\mmu$.
Throughout the remainder of this section, $\mbl, \mbr, \mprojl$, and $\mprojr$ are defined consistently with the choice of $\mblo$ and $\mbro$ in Definition~\ref{def:12}: $\mbl = \mbr = \id$, and $\mprojl$ and $\mprojr$ are the identity but with all but the first $r$ and $c$ 1's set to 0, respectively.
We next observe that this SVD commutes appropriately with exponentiated projections respecting the partition.

\begin{lemma}[Variant of {\cite[Lemma 14]{gslw18}}] \label{lem:14}
Let $\phi \in \R$. Following notation of \cref{def:12},
\begin{align*}
e^{\ii\phi(2\mprojl - \id)} = \begin{pmatrix} e^{\ii\phi}\id \\ & e^{-\ii\phi}\id \end{pmatrix},\; 
e^{\ii\phi(2\mprojr - \id)} = \begin{pmatrix} e^{\ii\phi}\id \\ & e^{-\ii\phi}\id \end{pmatrix},
\end{align*}
with appropriate block sizes,\footnote{
    These block sizes are such that the blocks of the product $e^{\ii \phi(2\mpi_\plef - \id)} \mmu e^{\ii \varphi(2\mpi_\prig - \id)}$ match.
    For example, if the top-left block of $\mmu$ is $r \times c$, then the top-left block of $e^{\ii \phi(2\mpi_\plef - \id)}$ is $r \times r$ and the top-left block of $e^{\ii \phi(2\mpi_\plef - \id)}$ is $c \times c$.
} and
\begin{align*}
    \begin{pmatrix} e^{\ii\phi}\id \\ & e^{-\ii\phi}\id \end{pmatrix}
    \begin{pmatrix} \mv_{1} \\ & \mv_{2} \end{pmatrix}
 &= \begin{pmatrix} \mv_{1} \\ & \mv_{2} \end{pmatrix}
    \begin{pmatrix} e^{\ii\phi}\id \\ & e^{-\ii\phi}\id \end{pmatrix}, \\
    \begin{pmatrix} \mw_{1} \\ & \mw_{2} \end{pmatrix}
    \begin{pmatrix} e^{\ii\phi}\id \\ & e^{-\ii\phi}\id \end{pmatrix}
 &= \begin{pmatrix} e^{\ii\phi}\id \\ & e^{-\ii\phi}\id \end{pmatrix}
    \begin{pmatrix} \mw_{1} \\ & \mw_{2} \end{pmatrix}.
\end{align*}
\end{lemma}

We next state our main technical claim, whose proof is deferred to the end of the section.

\begin{lemma} \label{lem:qsp-to-qsvt}
    Let $\Phi \in \mathbb{R}^{n+1}$ be the sequence of phase factors implementing the degree-$n$ polynomial $p(x) \in \mathbb{C}[x]$ via quantum signal processing (\cref{def:qsp-achievable}).
    Then we can compute the corresponding QSVT circuit (\cref{def:15}) applied to the following partitioned unitaries.
    First, the unitary associated with zero singular values.
    \begin{equation}
        \Big[\begin{pmatrix}
            \mzero^{r\times c} & \id_{r} \\ \id_{c} & \mzero^{c\times r}
        \end{pmatrix}\Big]_\Phi = \begin{pmatrix}
            p(0) \id_c & \cdot \\ \cdot & \cdot
        \end{pmatrix} \text{ for $n$ even and } \begin{pmatrix}
            \mzero^{r\times c} & \cdot \\ \cdot & \cdot
        \end{pmatrix} \text{ for $n$ odd.}
        \label{eq:block-xzero}           
    \end{equation}
    Next, the unitary associated with one singular values.
    \begin{equation}
        \Big[\begin{pmatrix}
            \id_r & \mzero^{r \times c} \\ \mzero^{c \times r} & -\id_c
        \end{pmatrix}\Big]_\Phi = \begin{pmatrix}
        p(1)\id_r & \cdot \\ \cdot & \cdot
        \end{pmatrix}. \label{eq:block-xone}
    \end{equation}
    Finally, the unitary for intermediate singular vaues: let $\mc, \ms \in \mathbb{C}^{r \times r}$ be diagonal with $\mc^2 + \ms^2 = \id$.
    \begin{equation}
        \Big[\begin{pmatrix}
            \mc & \hphantom{-}\ms \\ \ms & -\mc
        \end{pmatrix}\Big]_\Phi = \begin{pmatrix}
            p^{\sv}(\mc) & \cdot \\ \cdot & \cdot
        \end{pmatrix}. \label{eq:block-xmid}
    \end{equation}
\end{lemma}

Using this lemma, our main QSVT result (Theorem~\ref{thm:qsp-to-qsvt}) in the setting of Definition~\ref{def:12} follows.
\begin{proof}[Proof of \cref{thm:qsp-to-qsvt}, special case]
We recall the definition of $\mmu_\Phi$:
\begin{align}
    \mmu_\Phi &= \begin{cases}
        e^{\ii\phi_0(2\mpi_\plef-\id)} \mmu e^{\ii\phi_1(2\mpi_\prig-\id)} \prod_{j=1}^{\frac{n - 1}{2}} \mmu^\dagger e^{\ii\phi_{2j}(2\mpi_\plef - \id)} \mmu e^{\ii\phi_{2j+1}(2\mpi_\prig - \id)} &\text{if } n \text{ is odd, and} \\
        \hspace{5.8em} e^{\ii\phi_0(2\mpi_\plef-\id)} \prod_{j=1}^{\frac n 2} \mmu^\dagger e^{\ii\phi_{2j-1}(2\mpi_\plef - \id)}\mmu e^{\ii\phi_{2j}(2\mpi_\prig - \id)}  &\text{if } n \text{ is even.} \\
    \end{cases} \nonumber
\intertext{Using that $\mv$ and $\mw^\dagger$ from the CS decomposition $\mmu = \mv \md \mw^\dagger$ commute with their adjacent exponentiated reflections (\cref{lem:14}), we continue:}
    &= \begin{cases}
        \mv e^{\ii\phi_0(2\mpi_\plef-\id)} \md e^{\ii\phi_1(2\mpi_\prig-\id)} \parens[\Big]{\prod_{j=1}^{\frac{n - 1}{2}} \md^\dagger e^{\ii\phi_{2j}(2\mpi_\plef - \id)} \md e^{\ii\phi_{2j+1}(2\mpi_\prig - \id)}} \mw^\dagger &\text{if } n \text{ is odd, and} \\
        \hspace{5.8em} \mw e^{\ii\phi_0(2\mpi_\plef-\id)} \parens[\Big]{\prod_{j=1}^{\frac n 2} \md^\dagger e^{\ii\phi_{2j-1}(2\mpi_\plef - \id)} \md e^{\ii\phi_{2j}(2\mpi_\prig - \id)}} \mw^\dagger &\text{if } n \text{ is even.} \\
    \end{cases} \nonumber \\
    &= \begin{cases}
        \mv \md_\Phi \mw^\dagger &\text{if } n \text{ is odd, and} \\
        \mw \md_\Phi \mw^\dagger &\text{if } n \text{ is even.} \\
    \end{cases}\label{eq:uphi_simplify}
\end{align}
This reduces the problem to computing $\md_\Phi$.
Recall from \cref{rmk:cs-rotation} that the structure of $\md$ is
\begin{equation*}
    %\begin{pmatrix} \md_{11} & \md_{12} \\ \md_{21} & \md_{22} \end{pmatrix}
    \left(\begin{array}{@{}c|c@{}}
    \md_{11} &
    \md_{12} \\
    \hline
    \md_{21} &
    \md_{22} \end{array}
    \right)
    =
    \left(\begin{array}{@{}c|c@{}}
    \begin{matrix}
    \hla{\mzero} & & \\
    & \hlb{\mc} & \\
    & & \hlc{\id}
    \end{matrix} &
    \begin{matrix}
    \hla{\id} & & \\
    & \hphantom{-}\hlb{\ms} & \\
    & & \hphantom{-}\hlc{\mzero}
    \end{matrix} \\
    \hline
    \begin{matrix}
    \hla{\id} & & \\
    & \hlb{\ms} & \\
    & & \hlc{\mzero}
    \end{matrix} &
    \begin{matrix}
    \hla{\mzero} & & \\
    & \hlb{-\mc} & \\
    & & \hlc{-\id}
    \end{matrix}
    \end{array}
    \right)
    = \underbrace{\hla{\begin{pmatrix}
        \mzero & \id \\ \id & \mzero
    \end{pmatrix}}}_{\hla{\vsina \to \vsouta}}
    \oplus \underbrace{\hlb{\begin{pmatrix}
        \mc & \hphantom{-}\ms \\ \ms & -\mc
    \end{pmatrix}}}_{\hlb{\vsinb \to \vsoutb}}
    \oplus \underbrace{\hlc{\begin{pmatrix}
        \id & \hphantom{-}\mzero \\ \mzero & -\id
    \end{pmatrix}}}_{\hlc{\vsinc \to \vsoutc}}.
\end{equation*}
Similarly, for $\phi \in \R$,
\begin{align*}
    e^{\ii\phi(2\mpi_\plef - \id)}
    &= \left(\begin{array}{@{}c|c@{}}
    e^{\ii\phi}\id
    \\ \hline
    & e^{-\ii\phi}\id
    \end{array}\right)
    = \underbrace{\hla{\begin{pmatrix}
        e^{\ii\phi}\id \\ & e^{-\ii\phi}\id
    \end{pmatrix}}}_{\hla{\vsouta \to \vsouta}}
    \oplus \underbrace{\hlb{\begin{pmatrix}
        e^{\ii\phi}\id \\ & e^{-\ii\phi}\id
    \end{pmatrix}}}_{\hlb{\vsoutb \to \vsoutb}}
    \oplus \underbrace{\hlc{\begin{pmatrix}
        e^{\ii\phi}\id \\ & e^{-\ii\phi}\id
    \end{pmatrix}}}_{\hlc{\vsoutc \to \vsoutc}}, \\
    e^{\ii\phi(2\mpi_\prig - \id)}
    &= \left(\begin{array}{@{}c|c@{}}
    e^{\ii\phi}\id
    \\ \hline
    & e^{-\ii\phi}\id
    \end{array}\right)
    = \underbrace{\hla{\begin{pmatrix}
        e^{\ii\phi}\id \\ & e^{-\ii\phi}\id
    \end{pmatrix}}}_{\hla{\vsina \to \vsina}}
    \oplus \underbrace{\hlb{\begin{pmatrix}
        e^{\ii\phi}\id \\ & e^{-\ii\phi}\id
    \end{pmatrix}}}_{\hlb{\vsinb \to \vsinb}}
    \oplus \underbrace{\hlc{\begin{pmatrix}
        e^{\ii\phi}\id \\ & e^{-\ii\phi}\id
    \end{pmatrix}}}_{\hlc{\vsinc \to \vsinc}}.
\end{align*}
Leveraging this direct sum decomposition of $\md$, applying \cref{lem:qsp-to-qsvt} to each block yields
\begin{align*}
    \md_\Phi &= \hla{\Big[\begin{pmatrix}
        \mzero & \id \\ \id & \mzero
    \end{pmatrix}\Big]_\Phi}
    \oplus \hlb{\Big[\begin{pmatrix}
        \mc & \hphantom{-}\ms \\ \ms & -\mc
    \end{pmatrix}\Big]_\Phi}
    \oplus \hlc{\Big[\begin{pmatrix}
        \id & \hphantom{-}\mzero \\ \mzero & -\id
    \end{pmatrix}\Big]_\Phi} \\
    &= \begin{cases}
    \hla{\begin{pmatrix}
        \mzero & \cdot \\ \cdot & \cdot
    \end{pmatrix}}
    \oplus \hlb{\begin{pmatrix}
        p^\sv(\mc) & \cdot \\ \cdot & \cdot
    \end{pmatrix}}
    \oplus \hlc{\begin{pmatrix}
        p(1)\id & \cdot \\ \cdot & \cdot
    \end{pmatrix}} & \text{if }n\text{ is odd, and} \\
    \hla{\begin{pmatrix}
        p(0)\id & \cdot \\ \cdot & \cdot
    \end{pmatrix}}
    \oplus \hlb{\begin{pmatrix}
        p^\sv(\mc) & \cdot \\ \cdot & \cdot
    \end{pmatrix}}
    \oplus \hlc{\begin{pmatrix}
        p(1)\id & \cdot \\ \cdot & \cdot
    \end{pmatrix}} & \text{if }n\text{ is even.}
    \end{cases}
\end{align*}
So, for $n$ odd, recalling \eqref{eq:uphi_simplify} and $p(0) = 0$, we have
\begin{align*}
    \mpi_{\plef}\mmu_\Phi\mpi_{\prig}
    &= \mpi_{\plef}\mv \md_\Phi \mw^\dagger\mpi_{\prig} \\
    &= \begin{pmatrix}
        \id & \\ & \mzero
    \end{pmatrix}
    \begin{pmatrix}
        \mv_{1} \\ & \mv_{2}
    \end{pmatrix}
    \md_\Phi
    \begin{pmatrix}
        \mw_{1}^\dagger \\ & \mw_{ 2}^\dagger
    \end{pmatrix}
    \begin{pmatrix}
        \id & \\ & \mzero
    \end{pmatrix} = \begin{pmatrix}
        \mv_{1} & \\ & \mzero
    \end{pmatrix}
    \md_\Phi
    \begin{pmatrix}
        \mw_{1}^\dagger & \\ & \mzero
    \end{pmatrix} \\
    &= \left(\begin{array}{@{}c|c@{}}
    \mv_{1}\parens*{\begin{smallmatrix}
    \hla{\mzero} & & \\
    & \hlb{p^\sv(\mc)} & \\
    & & \hlc{p(1)\id}
    \end{smallmatrix}}\mw_{1}^\dagger &
    \begin{matrix}
    \hphantom{-} & \\
    & \hphantom{-}
    \end{matrix} \\
    \hline
    \begin{matrix}
    \hphantom{-} & \\
    & \hphantom{-}
    \end{matrix} &
    \begin{matrix}
    \hphantom{-} & \\
    & \hphantom{-}
    \end{matrix}
    \end{array}
    \right)
    = \left(\begin{array}{@{}c|c@{}}
    p^\sv(\ma) & \mzero \\
    \hline
    \mzero & \mzero
    \end{array}
    \right).
\end{align*}
Similarly, for $n$ even, we have
\begin{align*}
    \mpi_{\prig}\mmu_\Phi\mpi_{\prig}
    &= \mpi_{\prig}\mw \md_\Phi \mw^\dagger\mpi_{\prig} \\
    &= \begin{pmatrix}
        \mw_{1} & \\ & \mzero 
    \end{pmatrix}
    \md_\Phi
    \begin{pmatrix}
        \mw_{1}^\dagger & \\ & \mzero
    \end{pmatrix} \\
    &= \left(\begin{array}{@{}c|c@{}}
    \mw_{1}\parens*{\begin{smallmatrix}
    \hla{p(0)\id} & & \\
    & \hlb{p^\sv(\mc)} & \\
    & & \hlc{p(1)\id}
    \end{smallmatrix}}\mw_{ 1}^\dagger &
    \begin{matrix}
    \hphantom{-} & & \\
    & \hphantom{-}
    \end{matrix} \\
    \hline
    \begin{matrix}
    \hphantom{-} & \\
    & \hphantom{-}
    \end{matrix} &
    \begin{matrix}
    \hphantom{-} & \\
    & \hphantom{-}
    \end{matrix}
    \end{array}
    \right)
    = \left(\begin{array}{@{}c|c@{}}
    p^\sv(\ma) & \mzero \\
    \hline
    \mzero & \mzero
    \end{array}
    \right). \qedhere
\end{align*}
\end{proof}

We conclude the section by proving Lemma~\ref{lem:qsp-to-qsvt}.

\begin{proof}[Proof of \cref{lem:qsp-to-qsvt}]
The basic intuition behind this argument is that, by assumption and \eqref{eq:qsp},
\begin{align*}
    \begin{pmatrix}
        e^{\ii \phi_0} & 0 \\ 0 & e^{-\ii \phi_0}
    \end{pmatrix}
    \prod_{j=1}^n \begin{pmatrix}
            x & \sqrt{1-x^2} \\ \sqrt{1-x^2} & -x
    \end{pmatrix}
    \begin{pmatrix}
        e^{\ii \phi_j} & 0 \\ 0 & e^{-\ii \phi_j}
    \end{pmatrix}
    = \begin{pmatrix}
        p(x) & \cdot \\ \cdot & \cdot
    \end{pmatrix}.
\end{align*}
So, supposing we could evaluate the polynomial at a matrix $x \leftarrow \mc$, we get that
\begin{align*}
    ``\begin{pmatrix}
        e^{\ii \phi_0}\id &  \\  & e^{-\ii \phi_0}\id
    \end{pmatrix}
    \prod_{j=1}^n \begin{pmatrix}
            \mc & \sqrt{\id-\mc^2} \\ \sqrt{\id-\mc^2} & -\mc
    \end{pmatrix}
    \begin{pmatrix}
        e^{\ii \phi_j}\id & \\ & e^{-\ii \phi_j}\id
    \end{pmatrix}
    = \begin{pmatrix}
        p(\mc) & \cdot \\ \cdot & \cdot
    \end{pmatrix}."
\end{align*}
This should hold because block matrix multiplication operates by the same rules as scalar matrix multiplication, but requires care to handle the non-square case.
We formalize this argument in an elementary manner.
First, we consider \eqref{eq:block-xzero}.
Let $\mmu = \parens*{\begin{smallmatrix} \mzero & \id_{r} \\ \id_{c} & \mzero \end{smallmatrix}}$.
When $n$ is even,
\begin{align*}
    \mmu_\Phi &= \begin{pmatrix}
        e^{\ii \phi_{0}}\id  \\ & e^{-\ii \phi_{0}}\id
    \end{pmatrix}
    \prod_{j=1}^{\frac n 2}
    \begin{pmatrix}
        \mzero & \id \\ \id & \mzero
    \end{pmatrix}^\dagger
    \begin{pmatrix}
        e^{\ii \phi_{2j-1}}\id  \\ & e^{-\ii \phi_{2j-1}}\id
    \end{pmatrix}
    \begin{pmatrix}
        \mzero & \id \\ \id & \mzero
    \end{pmatrix}
    \begin{pmatrix}
        e^{\ii \phi_{2j}}\id  \\ & e^{-\ii \phi_{2j}}\id
    \end{pmatrix} \\
    &= \begin{pmatrix}
        e^{\ii \phi_{0}}\id  \\ & e^{-\ii \phi_{0}}\id
    \end{pmatrix}
    \prod_{j=1}^{\frac n 2}
    \begin{pmatrix}
        e^{\ii(\phi_{2j} - \phi_{2j-1})}\id & \mzero \\ \mzero & e^{-\ii(\phi_{2j} - \phi_{2j-1})}\id
    \end{pmatrix} \\
    &= \begin{pmatrix}
        e^{\ii\sum_{k=0}^{n}(-1)^{k}\phi_k}\id & \mzero \\ \mzero & e^{-\ii\sum_{k=0}^n(-1)^{k}\phi_k}\id
    \end{pmatrix}.
\end{align*}
Taking $\id$ and $\mzero$ to be $1$-dimensional scalars $1$ and $0$, this computation and Definition~\ref{def:qsp} also show that $p(0) = e^{\ii\sum_{k=0}^n(-1)^{k}\phi_k}$ yielding the desired conclusion. Similarly, when $n$ is odd,
\begin{align*}
    \mmu_\Phi &=
    \begin{pmatrix}
        e^{\ii \phi_{0}}\id  \\ & e^{-\ii \phi_{0}}\id
    \end{pmatrix}
    \begin{pmatrix}
        \mzero & \id \\ \id & \mzero
    \end{pmatrix}
    \begin{pmatrix}
        e^{\ii \phi_{1}}\id  \\ & e^{-\ii \phi_{1}}\id
    \end{pmatrix} \\
    &\hspace{8em}
    \prod_{j=1}^{\frac{n-1}{2}}
    \begin{pmatrix}
        \mzero & \id \\ \id & \mzero
    \end{pmatrix}^\dagger
    \begin{pmatrix}
        e^{\ii \phi_{2j}}\id  \\ & e^{-\ii \phi_{2j}}\id
    \end{pmatrix}
    \begin{pmatrix}
        \mzero & \id \\ \id & \mzero
    \end{pmatrix}
    \begin{pmatrix}
        e^{\ii \phi_{2j+1}}\id  \\ & e^{-\ii \phi_{2j+1}}\id
    \end{pmatrix} \\
    &= \begin{pmatrix}
        \mzero & e^{\ii\sum_{k=0}^n(-1)^{k}\phi_k}\id \\ e^{-\ii\sum_{k=0}^n(-1)^{k}\phi_k}\id & \mzero
    \end{pmatrix}.
\end{align*}
Next, we prove \eqref{eq:block-xone}.
Let $\mmu = \parens*{\begin{smallmatrix} \id_r & \mzero \\ \mzero & -\id_c \end{smallmatrix}}$.
Notice that $\mmu$ commutes with the other matrices in the expression $\mmu_\Phi$.
As an immediate consequence,
\begin{align*}
    \mmu_\Phi &= \begin{pmatrix}
        \id & \mzero \\ \mzero & (-1)^n\id
    \end{pmatrix}
    \prod_{k=0}^n \begin{pmatrix}
        e^{\ii \phi_k}\id & \mzero \\ \mzero & e^{-\ii \phi_k}\id
    \end{pmatrix}
    = \begin{pmatrix}
        e^{\ii \sum_{k=0}^n\phi_k}\id & \mzero \\ \mzero & (-1)^ne^{-\ii \sum_{k=0}^n\phi_k}\id
    \end{pmatrix}.
\end{align*}
As before, the same computation specialized to a $2$-dimensional $\mmu = \msigma_z$  shows that $p(1) = e^{\ii \sum_{k=0}^n\phi_k}$, giving the desired claim.
Finally to prove \eqref{eq:block-xmid}, let the diagonal entries of $\mc$ be $\{c_i\}_{i \in [r]}$. Then $\mmu$ is the direct sum of $r$ matrices of the form $\mr(c_i)$, where we recall we defined $\mr$ in Definition~\ref{def:qsp}. Applying Definition~\ref{def:qsp} to each $2 \times 2$ block, and comparing to Definition~\ref{def:16}, yields the conclusion.
\end{proof}

\subsection{Simplified QSVT in general bases}\label{ssec:gen_case}

We now finish the proof of Theorem~\ref{thm:qsp-to-qsvt} in the case of general bases $\mbl$, $\mbr$. For disambiguation we let $\mbl$, $\mbr$, $\mblo$, $\mbro$, $\mprojl$, $\mprojr$ refer to an arbitrary basis and associated subspace in Definition~\ref{def:block_encode}, and (for this section only) we let $\bmbl$, $\bmbr$, $\bmblo$, $\bmbro$, $\bmprojl$, $\bmprojr$ refer to the computational basis where $\bmblo$ and $\bmbro$ have the same number of columns as $\mblo$ and $\mbro$. These are related via
\begin{equation}\label{eq:convert_bases}
\mblo = \mbl \bmblo,\; \mprojl = \mbl \bmprojl \mbl^\dagger,\; \mbro = \mbr \bmbro,\; \mprojr = \mbr \bmprojr \mbr^\dagger.
\end{equation}
Finally, we define
\begin{equation}\label{eq:bmmudef}\bmmu \defeq \mbl^\dagger \mmu \mbr \iff \mmu = \mbl \bmmu \mbr^\dagger.\end{equation}
We prove the general case by reducing to the special case of Theorem~\ref{thm:qsp-to-qsvt} we proved in the prior Section~\ref{ssec:basic_case}, as suggested earlier. The following observation will be useful:
\begin{equation}\label{eq:compare_svd}
\bmprojl \bmmu \bmprojr = \mbl^\dagger \mprojl \mmu \mprojr \mbr.
\end{equation}

\begin{proof}[Proof of Theorem~\ref{thm:qsp-to-qsvt}, general case.]
For simplicity we will only prove the odd case, as the reduction in the even case is essentially identical. Recall that when $n$ is odd,
\begin{align*}
\mmu_\Phi &= e^{\ii\phi_1(2\mpi_\plef-\id)} \mmu \prod_{j\in[\frac{n - 1}{2}]} e^{\ii\phi_{2j}(2\mpi_\prig - \id)} \mmu^\dagger e^{\ii\phi_{2j+1}(2\mpi_\plef - \id)}\mmu \\
&= \mbl \Par{e^{\ii\phi_1(2\bmprojl-\id)} \bmmu \prod_{j\in[\frac{n - 1}{2}]} e^{\ii\phi_{2j}(2\bmprojr - \id)} \bmmu^\dagger e^{\ii\phi_{2j+1}(2\bmprojl - \id)}\bmmu} \mbr^\dagger \\
&= \mbl \begin{pmatrix}
p^\sv(\overline{\ma}) & \mzero \\ \mzero & \mzero
\end{pmatrix} \mbr^\dagger,\text{ where } \overline{\ma} \defeq \bmprojl \bmmu \bmprojr.
\end{align*}
In the second line, we used the definitions \eqref{eq:convert_bases}, \eqref{eq:bmmudef} and the fact that $\bmbl$, $\bmbr$ are unitary. Finally, in the third line we used the special case of Theorem~\ref{thm:qsp-to-qsvt} we proved earlier, applied to $\bmmu$. The conclusion follows by the claim
\begin{equation}\label{eq:compare_psv}\mbl \begin{pmatrix}
p^\sv(\overline{\ma}) & \mzero \\ \mzero & \mzero
\end{pmatrix} \mbr^\dagger = p^\sv(\ma),\text{ where } \ma = \mprojl \mmu \mprojr.\end{equation}
Indeed, as a consequence of \eqref{eq:compare_svd}, $\ma$ and $\overline{\ma}$ have the same singular values, their left singular vectors are related via rotation by $\mbl$, and their right singular vectors are related via rotation by $\mbr$. Comparing with the definition of $p^\sv$ from Definition~\ref{def:16} yields the claim \eqref{eq:compare_psv}.
\end{proof}

%!TEX root = main.tex
\section{QSVT and Chebyshev Series} \label{sec:cheb}

\begin{quotation}
    \textit{So it becomes tempting to look at approximation methods that go beyond interpolation, and to warn people that interpolation is dangerous... the trouble with this is that for almost all the functions encountered in practice, Chebyshev interpolation works beautifully!}
    \hspace{1em plus 1fill}---Trefethen, \cite{trefethen19}
\end{quotation}

\subsection{Chebyshev polynomials}\label{ssec:cheb_poly}

The QSVT framework gives a generic way of applying bounded polynomials to matrices. In applications of interest, the main goal is actually to apply a non-polynomial function; to capture these applications, it is important to develop tools for approximating the relevant functions with bounded polynomials. In this section, we introduce Chebyshev polynomials, our main tool for constructing approximations.
We present only properties which are needed to achieve our results.

\begin{definition}[Chebyshev polynomial]
    The degree-$n$ \emph{Chebyshev polynomial} (of the first kind), denoted $T_n(x)$, is the function that satisfies, for all $z \in \mathbb{C}$,
    \[
        T_n(\thalf(z + z^{-1})) = \thalf(z^n + z^{-n}).
    \]
\end{definition}

For $z = \exp(\ii \theta)$ for $\theta \in [-\pi, \pi]$ we may identify $x \defeq \half(z + z^{-1})$ for $x = \cos\theta$. This identification yields another familiar definition of the Chebyshev polynomials,
\[T_n(\cos(\theta)) = \cos(k\theta). \]
From these definitions we have that $\norm{T_n(x)}_{[-1,1]} \leq 1$, and that $T_n$ has the same parity as $n$, i.e.\ $T_n(-x) = (-1)^n T_n(x)$.
    We can invert $x = \thalf(z + z^{-1})$ to obtain $z = x \pm \sqrt{x^2 - 1}$.
    Consequently,
    \begin{align*}
        T_n(x) = \half\Par{\Par{x + \sqrt{x^2 - 1}}^n + \Par{x - \sqrt{x^2 - 1}}^{n}}.
    \end{align*}

\begin{definition}[Chebyshev coefficients]
    Let $f: [-1,1] \to \mathbb{C}$ be Lipschitz (i.e.\ $\abs{f(x) - f(y)} \leq C\abs{x-y}$ for finite $C$).
    Then $f$ has a unique decomposition into Chebyshev polynomials
    \begin{align*}
        f(x) = \sum_{k=0}^\infty a_k T_k(x),
    \end{align*}
    where the \emph{Chebyshev coefficients} $a_k$ absolutely converge, and can be computed via the following integral (counterclockwise) around the complex unit circle, where $\pi \ii$ is replaced by $2\pi \ii$ when $k = 0$:
    \begin{align} \label{eq:cheb-integral}
        a_k = \frac{1}{\pi\ii}\int_{\abs{z} = 1} z^{k-1}f(\thalf(z + z^{-1})) \dd z.
    \end{align}
\end{definition}
This comes from the Cauchy integral formula: $f(\half(z + z^{-1}))$ maps $[-1,1]$ (twice) onto the unit circle, so that when we write $f$ as a Laurent series $\sum_{k \in \Z} b_k(z^k + z^{-k})$, the coefficients $b_k$ match up with the coefficients $a_k$ (up to a factor of two). For more details, see Theorem 3.1 of \cite{trefethen19}. We will typically construct polynomial approximations via \emph{Chebyshev truncation}, defined as follows.

\begin{definition}[Chebyshev truncation]
    For a function $f: [-1,1] \to \mathbb{C}$ written as a Chebyshev series $f(x) = \sum_{k=0}^\infty a_k T_k(x)$, we denote the degree-$n$ \emph{Chebyshev truncation} of $f$ as
    \begin{align*}
        f_n(x) = \sum_{k=0}^n a_kT_k(x).
    \end{align*}
\end{definition}

To construct polynomial approximations time-efficiently, we should avoid computing the integral \eqref{eq:cheb-integral}.
So, instead of the Chebyshev truncation, we would actually compute the \emph{Chebyshev interpolant}~\cite[Chapter 3]{trefethen19}, the degree-$n$ polynomial that agrees with $f$ at the $n+1$ Chebyshev points $\cos(\frac{j\pi}{n})$ for $0 \le j \le n$.
Computing this interpolant only requires $n+1$ evaluations of $f$ to specify, and $O(n\log(n))$ additional time to compute its Chebyshev coefficients, via a fast cosine transform.
Chebyshev truncation and Chebyshev interpolation are closely related through standard bounds (see Eq.\ 4.9, \cite{trefethen19}); we focus on the latter as it is conceptually cleaner, but all bounds we prove extend to interpolation up to a factor of two.\footnote{This justifies our use of the quote from \cite{trefethen19} at the beginning of Section~\ref{sec:cheb}: though it is about Chebyshev interpolation rather than truncation, the spirit is the same.}

\subsection{Chebyshev series for standard functions}

If the function $f$ one wishes to approximate is standard, closed forms of the Chebyshev coefficients may be known, so one can take a Chebyshev truncation and explicitly bound the error:
\begin{align*}
    \norm{f - f_n}_{[-1,1]}
    = \norm[\Big]{\sum_{k=n+1}^\infty a_kT_k(x)}_{[-1,1]}
    \leq \sum_{k=n+1}^\infty \abs{a_k} \norm{T_k(x)}_{[-1,1]}
    = \sum_{k=n+1}^\infty \abs{a_k}.
\end{align*}
In other words, by choosing $n$ such that the coefficient tail sum is bounded by $\eps$, we obtain an $\eps$-uniform approximation on $[-1, 1]$. We list some standard Chebyshev coefficient series here, which can help for converting a Taylor or Fourier series to a Chebyshev series.
The notation $\sum^{\prime}$ means that if a term exists for $T_0$ in the summation, it is halved:
\begin{align*}
    x^m &= 2^{1-m}\sum_{n=0}^{\lfloor m/2\rfloor}{}^{\prime} \binom{m}{n} T_{m-2n}(x). \tag*{\cite[Eq.\ 2.14]{mh02}} \\
    e^{tx} &= 2\sum_{n=0}^{\infty}{}^{\prime} I_n(t) T_n(x). \tag*{\cite[Eq.\ 5.18]{mh02}} \\
    \sinh(tx) &= 2\sum_{n=0}^{\infty} I_{2n+1}(t) T_{2n+1}(x). \tag*{\cite[Eq.\ 5.19]{mh02}} \\
    \cosh(tx) &= 2\sum_{n=0}^{\infty}{}^{\prime} I_{2n}(t) T_{2n}(x). \tag*{\cite[Eq.\ 5.20]{mh02}}
\end{align*}
In the above display, $I_n$ denotes the modified Bessel function of the first kind. This function is typically defined as the solution to a differential equation, but for us it suffices to define $I_n(t)$ as ``the $n^{\text{th}}$ Chebyshev coefficient of $e^{tx}$,'' so by \eqref{eq:cheb-integral} (and pairing Laurent coefficients),
\begin{align*}
    I_n(t) \defeq \frac 1 {2\pi\ii}\oint_{\abs{z} = 1} z^{n-1} e^{\frac{t}{2}(z + z^{-1})}\dd z = \frac 1 {2\pi\ii}\oint_{\abs{z} = 1} z^{-n-1} e^{\frac{t}{2}(z + z^{-1})}\dd z.
\end{align*}

In the remainder of the section, to demonstrate the ``direct coefficient bound'' style of approximation error analysis, we analyze the use of Chebyshev truncation to give a uniform approximation to $f(x) = e^{tx}$ for $t \in \R$ on the interval $[-1, 1]$. The main result of this section is the following.

\begin{theorem}[{\cite[Theorem 4.1]{sv14}}, {\cite[Lemmas 57, 59]{gslw18}}]\label{thm:expbound}
    Let $\eps > 0$ and let $p(x)$ be the degree-$n$ Chebyshev truncation of $e^{tx}$ for $t \in \R$. 
    Then $\norm{p(x) - e^{tx}}_{[-1,1]} \leq \eps$ for
    \begin{align*}
        n \eqsim \begin{cases}
        \abs{t} + \frac{\log(1/\eps)}{\log(e + \log(1/\eps)/\abs{t})} & \eps \leq 1 \\
        \sqrt{\abs{t}\log(e^{\abs{t}}/\eps)} & \eps > 1
        \end{cases}.
    \end{align*}
\end{theorem}

To elaborate on what this bound states, there are four regimes of $\eps$.
\begin{enumerate}
    \item When $\eps \geq e^{\abs{t}}$, the zero polynomial $p(x) \equiv 0$ suffices.
    \item When $1 \leq \eps < e^{\abs{t}}$, we have $n \eqsim \sqrt{\abs{t}\log(e^{\abs{t}}/\eps)}$, or $n \eqsim \sqrt{\abs{t}\log(1/\delta)}$, rewriting $\eps = \delta e^{\abs{t}}$ to scale with the maximum value of $e^{\abs{t}x}$ on $[-1,1]$.
    This is the bound shown in \cite{sv14}.
    \item When $e^{-C\abs{t}} \leq \eps < 1$ for some universal constant $C$, the scaling is $n \eqsim \abs{t}$.
    \item When $\eps < e^{-C\abs{t}}$, the scaling is $\frac{\log(1/\eps)}{\log(e + \log(1/\eps)/\abs{t})}$.
    This is the bound shown in \cite{gslw18}.
\end{enumerate}
Theorem~\ref{thm:expbound} was first proven by combining results from \cite{sv14} and \cite{gslw18}. A recent work \cite{AggarwalA22} obtained the same upper bound, as well as a lower bound showing Theorem~\ref{thm:expbound} is tight. 
The bounds from \cite{sv14, gslw18} are each loose in certain regimes (\cite{sv14}'s bound, $\sqrt{\abs{t}\log(e^{\abs{t}}/\eps)} + \log(e^{\abs{t}}/\eps)$, is loose in regime 4, whereas \cite{gslw18} assumes $\eps < 1$), potentially due to the different proof techniques employed.
Specifically, while \cite{gslw18} proceeds by a standard Bessel function inequality to bound the tail terms of the Chebyshev truncation, \cite{sv14} proceeds by approximating monomials in the Taylor expansion of $e^{tx}$ with Chebyshev truncation.
As noted by \cite{LowC17}, this strategy bounds the tail terms of the Chebyshev truncation by an easier-to-understand series that dominates it.

We give another (arguably more straightforward) proof of Theorem~\ref{thm:expbound}, by bounding the error of Chebyshev truncation (a strategy also employed by \cite{AggarwalA22}).
To achieve the right bound in the $\eps > 1$ regime, we require a sharper bound on $I_n$, the Chebyshev coefficients of $e^{tx}$.

\begin{lemma}[Carlini's formula] \label{carlini}
    For $t \in \mathbb{R}$,
    \begin{align*}
        \abs{I_n(t)} < \frac{\exp(\sqrt{t^2 + n^2})(\sqrt{(n/t)^2 + 1} - n/\abs{t})^n}{2(n^2 + t^2)^{1/4}}.
    \end{align*}
\end{lemma}

We contribute an independent proof of \cref{carlini} in \cref{app:carlini}, which was proven without using Bessel-style techniques in \cite{AggarwalA22}.
While we would guess that this bound is well-known, we could not find this statement in the literature.\footnote{
    Off-the-shelf bounds like Kapteyn's inequality \cite[Eq. 10.14.8]{DLMF} are not quite tight enough for our purposes, since a very fine degree of control is necessary in the challenging regime where none of $n$, $t$, or $t/n$ remain constant.
}
The equivalent bound on the (unmodified) Bessel function of the first kind $J_k$ is due to Carlini~\cite[Chapter 1.4]{Watson44}, and can be viewed as a ``real-valued'' analog of \cref{carlini} following the equivalence $I_k(t) = \ii^{-k} J_k(\ii t)$.\footnote{
    Note that the real version of this statement is perhaps more non-trivial, since then the terms in the power series for the Bessel function are no longer nonnegative.
    Qualitatively similar statements may have also been made by Laplace, but for this claim Watson cites a book of the M\'ecanique C\'eleste without an English translation~\cite[page 7]{Watson44}.
    This felt like a good place to stop our investigation of Bessel function bounds.
}
Our proof follows \cite{Watson17} (who handled the real-valued version), and begins with a representation of a Bessel function as a contour integral. We bound this integral via the \emph{method of steepest descent}, where the contour is changed to a real-valued one, using that the integrand is analytic. Using Lemma~\ref{carlini}, we now prove Theorem~\ref{thm:expbound}.

\begin{proof}[Proof of Theorem~\ref{thm:expbound}]
By symmetry it suffices to take $t \geq 0$. We split into cases based on $\eps$. 

\textit{Case 1: $\eps \le 1$.} Define $r(t, \eps)$ to be the value $r$ such that $\eps = (t/r)^r$. We choose
\[n = \left\lceil r\Par{3t, \eps} \right\rceil.\]
Note that the function $(t/r)^r$ is decreasing for $r \ge t$ and hence in the regime $\eps \le 1$, we have $n \ge 3t$.
Then, recalling that $p$ is the degree-$n$ Chebyshev truncation, we have by Lemma~\ref{carlini} that
\begin{align*}
    \norm{p(x) - e^{tx}}_{[-1,1]}
    &\leq 2\sum_{k=n+1}^\infty \abs{I_k(t)} \le 2\sum_{k = n}^\infty |I_k(t)| \\
    &\leq 2\sum_{k=n}^\infty \frac{\exp(\sqrt{t^2 + k^2})(\sqrt{(k/t)^2 + 1} - k/t)^k}{2(k^2 + t^2)^{1/4}}\\
    &\leq \sum_{k=n}^\infty \exp(\sqrt{t^2 + k^2})(\sqrt{(k/t)^2 + 1} - k/t)^k \\
    &\leq \sum_{k=n}^\infty \exp(k\sqrt{(t/n)^2 + 1})(\sqrt{(n/t)^2 + 1} - n/t)^k \tag*{by $k \ge n$ and since $\sqrt{x^2 + 1} - x$ decreases in $x$}\\
    &\leq \sum_{k=n}^\infty \Big(\exp(\sqrt{(t/n)^2 + 1}) \cdot \frac{t}{2n}\Big)^k \tag*{by $\sqrt{1-x^2} - x \leq \frac 1 {2x}$}\\
    &\leq \sum_{k=n}^\infty \Big(\exp(\sqrt{10/9})\cdot \frac{t}{2n}\Big)^k \tag*{by $n \ge 3t$} \\
    &\leq \sum_{k=n}^\infty \Big(\frac{3t}{2n}\Big)^k
    \leq \Big(\frac{3t}{n}\Big)^{n} \sum_{k=1}^\infty \frac{1}{2^k} \leq \eps
    \tag*{by $n \geq r(3t, \eps)$.}
\end{align*}
The desired bound on $n$ in this regime then follows from \cite[Lemma 59]{gslw18} which shows that, for $\eps \in (0,1)$ and $t > 0$, $r(t, \eps) = \Theta(t + \frac{\log(1/\eps)}{\log(e + \log(1/\eps)/t)})$. 

\textit{Case 2: $\eps > 1$.}
For $\eps \in (1, 2]$ the conclusion follows from our proof when $\eps = 1$, so assume $\eps > 2$;\footnote{If $t$ is a sufficiently small constant, then applying the $\eps = 1$ case gives a constant-degree polynomial. Otherwise, $t$ is sufficiently large to outweigh constant-factor changes in $\eps$ (and hence additive changes in $\log \frac 1 \eps$).} for $\eps \ge e^t$ the zero polynomial suffices so assume $\eps < e^t$. Let $\delta = \frac{\eps - 1} {5}$ and $m = \lceil 3t \rceil$. We choose
\[n = \left\lceil \sqrt{100t\Par{t + \log \frac 1 \delta}} \right\rceil \ge 10\sqrt{t},\]
where we use $t + \log \frac 1 \delta \ge 1$ for our range of $\eps$.
The claim then follows from $\delta = \Theta(\eps)$ and combining:
\begin{equation}\label{eq:splitatr}
\begin{aligned}
2\sum_{k = n + 1}^{m - 1} |I_k(t)| \le 5\delta,\;
2\sum_{k = m}^\infty |I_k(t)| \le 1.
\end{aligned}
\end{equation}
The second claim in \eqref{eq:splitatr} was already shown by our earlier derivation setting $\eps = 1$, since $m \ge r(3t, 1) = 3t$. For bounding the first sum, the following estimate will be helpful: for $k \le 3t$,
\begin{equation}\label{eq:besselhelper}
\begin{aligned}
\exp\Par{\sqrt{t^2 + k^2}}\Par{\sqrt{\Par{\tfrac k t}^2 + 1} - \tfrac k t}^k &\le \exp\Par{\sqrt{t^2 + k^2} + k\Par{\sqrt{\Par{\tfrac k t}^2 + 1} - \tfrac k t - 1}} \\
&= \exp\Par{t\Par{1 + \tfrac k t}\Par{\sqrt{\Par{\tfrac k t}^2 + 1} - \tfrac k t}} \\
&\le \exp\Par{t\Par{1 - 0.01\Par{\tfrac k t}^2}}.
\end{aligned}
\end{equation}
The last equation used $(1 + x)(\sqrt{1 + x^2} - x) \le 1 - 0.01x^2$ for $0 \le x \le 3$. Since $m - 1 \le 3t$,
\begin{align*}
\sum_{k = n + 1}^{m - 1} |I_k(t)| &\leq \frac{1}{2(n^2 + t^2)^{1/4}}\sum_{k=n+1}^{m - 1} \exp\Par{t\Par{1 - 0.01\Par{\frac k t}^2}} \\
    &= \frac{e^t}{2(n^2 + t^2)^{1/4}}\sum_{k=n+1}^{m - 1} \exp\Par{-\frac{k^2}{100t}} \\
    &\le \frac{e^t}{2(n^2 + t^2)^{1/4}}\int_{n}^\infty \exp\Par{-\frac{x^2}{100t}}\dd x \\
    &= \frac{e^t \sqrt{25t}}{(n^2 + t^2)^{1/4}} \int_{\frac{n}{\sqrt{100t}}}^\infty \exp\Par{-x^2} \dd x \\
    &\le \frac{e^t \sqrt{25 t}}{(n^2 + t^2)^{1/4}} \cdot \Par{\half\exp\Par{-\frac{n^2}{100t}}} \le \frac 5 2\exp\Par{t - \frac{n^2}{100t}} \le \frac 5 2\delta.
\end{align*}

The first line used \cref{carlini} and \eqref{eq:besselhelper}, and the second-to-last used the Gaussian tail bound \eqref{eq:erf-bound}:
\[\int_{\frac{n}{\sqrt{100t}}}^\infty \exp\Par{-x^2} \dd x < \exp\Par{-\frac{n^2}{100t}} \cdot \frac{1}{1 + \frac n {\sqrt{100t}}} \le \half \exp\Par{-\frac{n^2}{100t}},\]
where we used $n \ge 10\sqrt{t}$ and $\exp(t - \frac{n^2}{100t}) \le \delta$ by construction.
\end{proof}

\subsection{Bounded approximations via Chebyshev series: a user's guide}\label{sec:bounded}

Two issues often arise when using polynomial approximations for QSVT.
First, we may not know explicitly what the Chebyshev coefficients of our desired function are.
Second, even when we do, Chebyshev truncation may be bad for our purposes, since our criteria is \emph{different} from uniform approximation on $[-1, 1]$.
For example, quantum linear systems requires a polynomial approximation close to $x^{-1}$ on $[-1, -\tfrac 1 \kappa] \cup [\tfrac 1 \kappa, 1]$, but it merely needs to be bounded on $[-\tfrac 1 \kappa, \tfrac 1 \kappa]$.
This bounded requirement is necessary to use the machinery of Section~\ref{sec:qsvt} (see \cref{rmk:qsvt-to-alg}).

Since quantum computing researchers are resourceful, we often see good polynomial approximations derived through ad hoc techniques to tame the function at poorly-behaved points.
For example, \cite{cks17} performs Chebyshev truncation on the polynomial $(1-(1-x^2)^b) \cdot \tfrac 1 x$ instead of $\tfrac 1 x$, and this has the desired properties.
However, as \cite{gslw18} points out, there are generic ways to find approximations to piecewise smooth functions which satisfy the ``$\eps$-close on smooth pieces, but bounded near points of discontinuity'' requirement of QSVT, with $\log\frac{1}{\eps}$ scaling in the degree.
Their proof of this claim assumes that the Taylor series coefficients of the smooth functions are bounded.

We give a variant of this result that only assumes that the smooth functions are bounded on ellipses in complex space, and has a self-contained proof based on Chebyshev series.
This alternate version matches \cite{gslw18} for sufficiently small $\eps$, and otherwise loses a logarithmic factor (though under a weaker assumption on the function to be approximated, see Remark~\ref{rem:bounded_comments}).
To do so, we combine a powerful meta-technique for bounding the Chebyshev coefficients of analytic functions with applications of explicit thresholding functions. This meta-technique is stated as Theorem~\ref{thm:trefethen}; in many prominent settings, a direct application already yields near-optimal polynomial approximations.

\begin{theorem}[{\cite[Theorems~8.1 and 8.2]{trefethen19}}] \label{thm:trefethen}
    Let $f$ be an analytic function in $[-1,1]$ and analytically continuable to the interior of the Bernstein ellipse $E_{\rho} = \{\frac12(z + z^{-1}) : \abs{z} = \rho\}$, where it satisfies $\abs{f(x)} \leq M$.
    Then its Chebyshev coefficients satisfy $\abs{a_0} \leq M$ and $\abs{a_k} \leq 2M\rho^{-k}$ for $k \geq 1$.
    Consequently, for each $n \geq 0$, its Chebyshev projections satisfy
    \begin{align*}
        \norm{f - f_n}_{[-1,1]} &\leq \frac{2M\rho^{-n}}{\rho - 1},
    \end{align*}
    and choosing $n = \lceil \frac{1}{\log(\rho)}\log\frac{2M}{(\rho-1)\eps} \rceil$, we have $\norm{f - f_n}_{[-1,1]} \leq \eps$.
\end{theorem}
\begin{proof}
Recall from \eqref{eq:cheb-integral} (and since inverting $z$ does not change the contour integral) that for $k \ge 1$,
\begin{align*}
a_k = \frac{1}{\pi \ii} \int_{|z| = 1} z^{-(k + 1)} f(\thalf(z + z^{-1})) \dd z.
\end{align*}
The boundary of $E_\rho$ is given by $\half(z + z^{-1})$ for $|z| = \rho$, and $f$ is analytic in $E_\rho$, so we may choose a different contour without affecting the value of the integral:
\begin{align*}
a_k = \frac{1}{\pi \ii} \int_{|z| = \rho} z^{-(k + 1)} f(\thalf(z + z^{-1})) \dd z.
\end{align*}
The conclusion follows from the facts that the circumference of $|z| = \rho$ is $2\pi \rho$ and the function is bounded by $M$. A similar argument gives the case $k = 0$, where \eqref{eq:cheb-integral} has $2\pi\ii$ in the denominator.
\end{proof}

Theorem~\ref{thm:trefethen} shows that if one can analytically continue $f$ to a Bernstein ellipse with $\rho = 1 + \alpha$ for small $\alpha$, then a degree $\approx \frac 1 \alpha$ polynomial obtains good approximation error on $[-1, 1]$. Unfortunately, since the approximation in Theorem~\ref{thm:trefethen} is based on Chebyshev truncation, the approximation rapidly blows up outside the range $[-1, 1]$ (in Lemma~\ref{lem:chebz-bound}, we give estimates on the growth of Chebyshev polynomials, i.e.\ that the $n^{\text{th}}$ polynomial grows as $O(|x|^n)$ for $x$ sufficiently outside $[-1, 1]$). 
In interesting applications of the QSVT framework, this is an obstacle.
For example, to use QSVT for matrix inversion, we need a polynomial approximation to $x^{-1}$ on $[\delta, 1]$ that is bounded on $[-1,1]$.
Upon linearly remapping $[\delta, 1]$ to $[-1, 1]$, this corresponds to a bounded approximation on $[-b, 1]$ for some $b > 1$, so Chebyshev truncations give us a very poor degree of control. 

To this end, we provide the following ``bounded approximation'' variant of Theorem~\ref{thm:trefethen}, as a user-friendly way of extending it to applications of the QSVT framework.

\begin{theorem}\label{thm:main_bounded}
    Let $f$ be an analytic function in $[-1, 1]$ and analytically continuable to the interior of $E_\rho$ where $\rho = 1 + \alpha$, where it is bounded by $M$.
    For $\delta \in (0,\frac 1 C \min(1, \alpha^2))$ where $C$ is a sufficiently large constant, $\eps \in (0, 1)$, and $b > 1$, there is a polynomial $q$ of degree $O(\frac b \delta \log\frac{b}{\delta\eps})$ such that
    \begin{align*}
        \norm{f - q}_{[-1, 1]} &\leq M\eps, \\
        \norm{q}_{[-(1 + \delta), 1 + \delta]} &\leq M, \\
        \norm{q}_{[-b, -(1 + \delta)] \cup [1 + \delta,b]} &\leq M\eps.
    \end{align*}
\end{theorem}
\begin{proof}[Proof sketch.]
We give a formal proof in Section~\ref{ssec:bounded_proof}, but briefly summarize our proof strategy here. 
\begin{enumerate}
    \item Applying Theorem~\ref{thm:trefethen} gives $f_n$ of degree $n \approx \frac 1 \alpha$ approximating $f$ in the interval $[-1, 1]$, but $f_n$ does not satisfy the other required conclusions due to its growth outside $[-1, 1]$.
    \item We multiply $f_n$ by a ``threshold'' $r$ based on the Gaussian error function $\erf$, whose tails decay much faster than the Chebyshev polynomials grow outside $[-1, 1]$. Our function $r$ has the property that inside $[-1, 1]$, it is close to $1$, and outside $[-(1 + \delta), 1 + \delta]$, it is close to $0$.
    \item Using bounds on the growth of $\erf$, we show $r \cdot f_n$ is bounded on a Bernstein ellipse of radius $1 + \frac \delta b$ appropriately rescaled, and applying Theorem~\ref{thm:trefethen} once more gives the conclusion.
\end{enumerate}  
The final proof requires some care to obtain the claimed scalings on the windows of approximation, but we include this tedium to make the theorem statement as simple to use as possible.
\end{proof}

\begin{remark}\label{rem:bounded_comments}
\cref{thm:main_bounded} is an alternative to Corollary 66 of \cite{gslw18}.
Translating the statement there to our setting, the polynomial approximation it would achieve has degree $O(\frac{b}{\delta}\log\frac{M}{\eps})$.
Our approximation of degree $O(\frac{b}{\delta}\log\frac{b}{\delta\eps})$ is comparable, matching when $\eps$ and $\frac \delta b$ are polynomially related.

We note that our Theorem~\ref{thm:main_bounded} has a $\log\frac{b}{\delta\eps}$ dependence (instead of $\log\frac{1}{\eps}$) because we use a slightly weaker type of assumption: not only does \cite[Corollary 66]{gslw18} assume that its function $f(x) = \sum_{k=0}^\infty a_k x^k$ is analytic and bounded by $M$ on a disk of radius $1 + \delta$, it also assumes that the Taylor series coefficients $\abs{a_k}$ satisfy $\sum_{k=0}^\infty \abs{a_k}(1+\delta)^k \leq M$.
Without this final condition, boundedness merely implies that $\abs{a_k} = O((1+\delta)^{-k})$, and this slight weakening leads to an additional logarithmic factor when $\eps$ is large.
In most applications, the difference is negligible; both strategies have an additional polynomial overhead on $\frac b \delta$, which typically dominates a $\log \frac b \delta$ dependence.

Finally, the precondition of Theorem~\ref{thm:main_bounded} is weaker than the requirement of Corollary 66 of \cite{gslw18} in another sense.
Specifically, \cite{gslw18} assumes a locally bounded Taylor series in a scaled unit circle in the complex plane, whereas we only require a bound on a (potentially much smaller) Bernstein ellipse, which could enable more applications.
\end{remark}

We use the rest of the section to provide a user's guide on applying Theorem~\ref{thm:main_bounded} to boundedly approximate various piecewise smooth functions. All of our applications proceed as follows.

\begin{enumerate}
    \item We linearly rescale the ``region of interest,'' i.e.\ the part of $\R$ where we wish to approximate a function via bounded polynomials, to the interval $[-1, 1]$.
    \item We apply Theorem~\ref{thm:main_bounded} to the rescaled function for appropriate choices of $b$ and $\delta$, so the region where the approximation must be bounded is captured upon undoing the rescaling.
    \item If additional properties of the bounded approximation are desired, e.g.\ a parity requirement, we use the additional implications of Theorem~\ref{thm:main_bounded} to obtain these properties.
\end{enumerate}

A simple application of \cref{thm:main_bounded} is obtaining degree-$O(\frac{1}{\delta}\log\frac{1}{\delta\eps})$ polynomial approximations to the sign and rectangle functions (where our guarantee is $\eps$-closeness outside of a $\delta$ interval around the points of discontinuity, as described in \cite[Lemmas 25 and 29]{gslw18}).\footnote{This result shows that direct Chebyshev truncation is sometimes not enough for these slightly different approximation guarantees: the sign function has Chebyshev series $\sum_{k\geq 0} \frac{4}{\pi}\frac{(-1)^k}{2k+1}T_{2k+1}(x)$~\cite[Exercise 3.6]{trefethen19}, which cannot be truncated without paying $\Omega(1)$ error.}
We leave this as an exercise. We begin with a bounded approximation to the rescaled exponential function in Corollary~\ref{cor:exp_bounded}. Such bounds have previously seen use in quantum applications of the multiplicative weights framework via QSVT, to design faster approximate solvers for linear programs \cite{vanApeldoornG19, BoulandGJSWT23}.

\begin{corollary}\label{cor:exp_bounded}
Let $\eps \in (0, 1)$, and let $f(x) = \exp(\beta x)$ for $\beta \ge 1$. There exists a polynomial $p$ of degree $O(\beta \log \frac \beta {\eps})$ such that $\norm{p}_{[-1, 1]} = O(1)$ and $\norm{p - f}_{[-1, 0]} \le \eps$.
\end{corollary}
\begin{proof}
First, we rescale so the region of interest is $[-1, 1]$: let $g(y) = f(\beta(\thalf (y - 1)))$. Note that $g(y)$ is analytic everywhere and bounded by a constant on $E_\rho$ for $\rho = 1 + \beta^{-\half}$. To see this, the magnitude of $\half(z - 1)$ for $z \in E_\rho$ is maximized when $z$ is furthest from $1$, and Fact~\ref{lem:bern-bounds} shows this magnitude is $O(\frac 1 \beta)$. Hence, applying Theorem~\ref{thm:main_bounded} with $b = 3$ and $\delta = \Theta(\frac 1 \beta)$ for a sufficiently small constant yields the claim upon shifting the region of interest back, since $\half(y - 1) = 1$ for $y = 3$.
\end{proof}

Next, in Corollary~\ref{cor:approx-arcsin} we provide an analog of Lemma 70 in \cite{gslw18}, regarding the bounded approximation of $\arcsin$, using the framework of Theorem~\ref{thm:main_bounded}.

\begin{corollary}\label{cor:approx-arcsin}
    Let $\delta, \eps \in (0,1)$, and let $f(x) = \frac 2 \pi \arcsin(x)$. There exists an odd polynomial $p(x)$ of degree $O(\frac{1}{\sqrt{\delta}}\log\frac{1}{\delta\eps})$ such that $\norm{p}_{[-1, 1]} \le 1$ and $\norm{p - f}_{[-(1 - \delta), 1 - \delta]} \le \eps$.
\end{corollary}
\begin{proof}
    First, we rescale so the region of interest is $[-1,1]$: let $\barcsin(x) = \arcsin((1-\delta)x)$.
    The $\arcsin$ function is analytic on $\mathbb{C} \setminus ((-\infty, 1] \cup [1, \infty))$, so we choose $\rho = 1 + \sqrt{2\delta}$ so that $\barcsin(x)$ is analytic on the interior of $E_\rho$ by the first bound in Fact~\ref{lem:bern-bounds}.
    By the maximum modulus principle, the maximum of $\overline{\arcsin}$ is achieved on the boundary of the ellipse.
    We can bound this using that, for $\abs{z} \leq 1$ (so the Taylor series~\cite[Eq.~4.24.1]{DLMF} converges),   \begin{equation}\label{eq:arcsin_bound}
        \abs{\arcsin{z}} = \left|\sum_{n = 0}^\infty \frac{(2 n)!} {2^{2 n} (n!)^2} \frac{z^{2 n + 1}}{2 n + 1}\right| \leq \sum_{n = 0}^\infty \frac{(2 n)!} {2^{2 n} (n!)^2} \frac{\abs{z}^{2 n + 1}}{2 n + 1} \leq \arcsin\abs{z} \leq \frac{\pi}{2}.
    \end{equation}
    We can further verify by Fact~\ref{lem:bern-bounds} that $|z| \le 1 + \delta$ for $z \in E_\rho$, so the above display yields
    \begin{align*}
    \left|\barcsin\Par{z}\right| &= \left|\arcsin\Par{(1 - \delta)z}\right| \le \frac \pi 2.
    \end{align*}
    So, by \cref{thm:main_bounded} with $b \gets \frac 1 {1 - \delta}$ and $\delta \gets \frac \delta C$ for a sufficiently large $C$, there is a polynomial $q$ with $\norm{q - \barcsin}_{[-1, 1]} \le \frac \pi 2 \eps$ and $\norm{q}_{[-(1 - \delta)^{-1}, (1 - \delta)^{-1}]} \le \frac \pi 2$. 
    Letting $p((1 - \delta)x) = \frac 2 \pi q(x)$, we have the desired bounds.
    The degree of $p$ is $O(\frac{1}{\delta}\log\frac{1}{\delta\eps})$, and it is odd by Corollary~\ref{cor:parity_bounded} as $\barcsin$ is odd.
\end{proof}

In Corollary~\ref{cor:approx-exptarcsin}, we further apply Theorem~\ref{thm:main_bounded} to the ``fractional query'' setting of Corollary 72 in \cite{gslw18}, which requires approximations to $\exp(\ii t \arcsin(x))$ for small $t$. As in \cite{gslw18}, we provide bounded approximations to $\cos(t\arcsin(x))$ and $\sin(t\arcsin(x))$ through our framework.

\begin{corollary}\label{cor:approx-exptarcsin}
    Let $\eps \in (0,1)$ and $t \in [-1,1]$. There exists an even polynomial $p$ and an odd polynomial $q$ of degree $O(\log \frac 1 \eps)$ such that $\norm{p}_{[-1, 1]} \le 1$, $\norm{q}_{[-1, 1]} \le 1$, and
    \begin{align*}
        \norm{p(x) - \cos(t \arcsin(x))}_{[-\half, \half]} \leq \eps,\;
        \norm{q(x) - \sin(t \arcsin(x))}_{[-\half, \half]} \leq \eps.
    \end{align*}
\end{corollary}
\begin{proof}
    First, we rescale so the region of interest is $[-1,1]$: let $f(x) = \cos(t \arcsin(\frac x 2))$ and $g(x) = \sin(t \arcsin(\frac x 2))$.
    These are analytic on $\mathbb{C} \setminus ((-\infty, -2] \cup [2, \infty))$, since that is where $\arcsin(\frac x 2)$ is analytic.
    Let $\rho = 2$, so $f$ and $g$ are analytic on the interior of $E_\rho$. We observe that for all $z \in \C$,
    \begin{align*}
    |\cos(z)| = \half\left|e^{\ii z} + e^{-\ii z}\right| \le \half |e^{\ii z}| + \half |e^{-\ii z}| \le \cosh(|z|), 
    \end{align*}
    as $\cosh$ is increasing and the imaginary part of $z$ is at most $|z|$. A similar argument shows $|\sin(z)| \le \cosh(|z|)$. By Fact~\ref{lem:bern-bounds} we observe that every point in the interior of $\half E_\rho$ has modulus $\le \frac 3 4$, and $|\arcsin|$ is bounded in this region by $\frac \pi 2$ (see \eqref{eq:arcsin_bound}), so for $z \in E_\rho$, $\left|f\Par{z}\right| = \left|\cos\Par{t\arcsin\Par{\frac z 2 }} \right| \le \cosh\Par{\frac \pi 2}$, 
    and we may analogously bound $g$ on $E_\rho$. Taking $b = 2$ and $\delta$ to be a sufficiently small constant in \cref{thm:main_bounded}, and rescaling the region of interest, gives the conclusion. The parities of $p$ and $q$ follow from Corollary~\ref{cor:parity_bounded} and the parities of $\cos(t\arcsin(x))$ and $\sin(t\arcsin(x))$.
\end{proof}

Finally, Corollary~\ref{cor:approx-x-c} gives a variant of Corollaries 67 and 69 in \cite{gslw18}, regarding the bounded approximation of negative power functions. Our bound has a slightly worse logarithmic factor in some regimes (as discussed in Remark~\ref{rem:bounded_comments}), but otherwise agrees with the bounds in \cite{gslw18} up to a constant factor, using arguably a more standard approach.

\begin{corollary} \label{cor:approx-x-c}
    Let $\delta, \eps \in (0, 1)$, and let $f(x) = \abs{\frac \delta x}^{c}$ for $c > 0$.
    There exist both even and odd polynomials $p(x)$ of degree $O(\frac{\max(1, c)}{\delta}\log\frac{1}{\delta\eps})$ such that $\norm{p}_{[-1, 1]} \le 3$ and $\norm{p - f}_{[\delta, 1]} \leq \eps$.
\end{corollary}
\begin{proof}
    Assume $\delta$ is sufficiently small, else taking a smaller $\delta$ only affects the bound by a constant. We rescale the region of interest: $x = \frac{1-\delta}{2}y + \frac{1+\delta}{2}$ is in $[\delta, 1]$ for $y \in [-1,1]$, so let
    \begin{align*}
        g(y) \defeq \delta^c\Big(\frac{1-\delta}{2}y + \frac{1+\delta}{2}\Big)^{-c}.
    \end{align*}
    We require a bound of $g$ on $E_\rho$ for $\rho = 1 + \sqrt{\delta/4\max(1, c)}$.
    Since $f$ is largest closest to the origin, $g$ is largest at the point closest to $-\frac{1+\delta}{1-\delta}$, i.e.\ $-\frac12(\rho + \rho^{-1}) > -(1 + \frac \delta {8\max(1, c)})$ by \cref{lem:bern-bounds}. Further,
    \begin{align*}
        g\Par{-\frac12(\rho + \rho^{-1})} &\leq g\Par{-\Par{1 + \frac \delta {8\max(1, c)}}} \\
        &\leq \delta^c\Big(-\frac{1-\delta}{2}\Par{1 + \frac \delta {8\max(1, c)}} + \frac{1+\delta}{2}\Big)^{-c} \\
        &= \Par{1 - \frac{1 - \delta}{16\max(1, c)}}^{-c} \leq \frac 3 2.
    \end{align*}
    Let $\tdelta = \frac \delta {4C\max(1, c)}$ for sufficiently large $C$, and $b = 4$. \cref{thm:main_bounded} yields $q(y)$ satisfying:
    \begin{align*}
        \norm{q(y) - g(y)}_{[-1,1]} \leq \eps, \; \norm{q(y)}_{[-(1 + \tdelta), 1+\tdelta]} \leq 2^c,\;
        \norm{q(y)}_{[-4, -(1 + \tdelta)] \cup [1+\tdelta, 4]} \leq \eps.
    \end{align*}
    Shifting back $y = \frac{2}{1-\delta}(x - \frac{1+\delta}{2})$, it is clear for sufficiently large $C$ that $y = -\frac{1 + 3\delta}{1 - \delta}$ (which corresponds to $x = -\delta$) has $y < -(1 + \tdelta)$, and $y = -\frac{3 + \delta}{1 - \delta}$ (which corresponds to $x = -1$) has $y > -4$. So,
    \begin{equation}\label{eq:qfbounds}
        \begin{aligned}
            \left\|q\Par{\frac{2}{1-\delta}\Par{x - \frac{1+\delta}{2}}} - f(x)\right\|_{[\delta,1]} \leq \eps, \\
            \left\|q\Par{\frac{2}{1-\delta}\Par{x - \frac{1+\delta}{2}}}\right\|_{[-\delta, \delta]} \leq 2^c, \\
            \left\|q\Par{\frac{2}{1-\delta}\Par{x - \frac{1+\delta}{2}}}\right\|_{[-1, -\delta]} \leq \eps.
        \end{aligned}
    \end{equation}
    Depending on whether we wish the final function to be even or odd, we take \[p(x) =  q\Par{\frac{2}{1-\delta}\Par{x - \frac{1+\delta}{2}}} \pm q\Par{\frac{2}{1-\delta}\Par{-x - \frac{1+\delta}{2}}}.\]
    Then the guarantees of \eqref{eq:qfbounds} give $    \norm{p(x) - f(x)}_{[\delta,1]} \leq 2\eps$ and $\norm{p(x)}_{[-1, 1]} \leq 3$, and we rescale $\eps$ to conclude.
    The final degree of the polynomial is the degree of $q(y)$: $O(\frac{\max(1, c)}{\delta}\log\frac{1}{\delta\eps})$.
\end{proof}

\subsection{Separating bounded and unbounded polynomial approximations}\label{ssec:separation}

In this section, we show that QSVT's requirement that the polynomials it implements be bounded can worsen the quality of approximations.
Specifically, we prove a simple separation result which shows that polynomial approximations may necessarily require larger degree under an additional boundedness constraint.
This follows from the observation that bounded degree-$d$ polynomials can have derivative as large as $d^2$ near the boundary of $[-1, 1]$, yet are bounded by $O(d)$ on the interior.
This is formalized by the following classical inequality due to Bernstein~\cite{bernstein12}.

\begin{proposition}[Theorem 2, \cite{schaeffer41}]\label{prop:bern-deriv-lb}
Let $p$ be a degree-$d$ polynomial with rational coefficients satisfying $|p(x)| \le 1$ for all $x \in [-1, 1]$. Then
\[d \ge |p'(x)| \sqrt{1 - x^2} \text{ for all } x \in [-1, 1].\]
\end{proposition}

Proposition~\ref{prop:bern-deriv-lb} is troublesome for obtaining the type of bound we want since it depends on derivatives of $p$, the approximation, rather than $f$, the function to be approximated. We next give a simple extension of Proposition~\ref{prop:bern-deriv-lb}, with degree lower bounds depending on a quantity  $\sup_{x, y}\frac{\abs{f(x) - f(y)} - 2\eps}{\abs{x-y}}$ which can be viewed as a ``robust'' Lipschitz constant of $f$.
For example, if $f$ is a differentiable function with derivative $\geq L$ on an interval of length at least $\frac{4\eps}{L}$, then this quantity is $\geq \frac L 2$, and taking $\eps \to 0$ recovers the maximum derivative of $f$.

\begin{proposition}\label{prop:generic-deriv-lb}
For $\Delta \in (0, 1]$ and $S \subset [-\Delta, \Delta]$, let $f(x): S \to [-1, 1]$ be a function with polynomial approximation $p(x)$ such that, for some approximation error $\eps > 0$,
\begin{align}
    \abs{p(x) - f(x)} &\leq \eps \text{ for all } x \in S,\text{ and} \label{eq:lower-approx}\\
    \abs{p(x)} &\leq 1 \text{ for all } x \in [-1, 1]. \label{eq:lower-mag}
\end{align}
Then
\begin{align*}
    \deg(p) \geq \sqrt{1-\Delta^2}\sup_{\substack{x, y \in S \\ x \neq y}}\frac{\abs{f(x) - f(y)} - 2\eps}{\abs{x-y}}.
\end{align*}
\end{proposition}
\begin{proof}
Consider some distinct $x, y \in S$, and let
\begin{align*}
    L = \frac{\abs{f(x) - f(y)} - 2\eps}{\abs{x - y}}.
\end{align*}
Then by \eqref{eq:lower-approx},
\begin{align*}
    \abs{p(x) - p(y)} \geq \abs{f(x) - f(y)} - 2\eps = L\abs{x - y},
\end{align*}
so by the intermediate value theorem, for some $\xi$ between $x$ and $y$, $\abs{p'(\xi)} \geq L$.
Since $p(x)$ is bounded by 1 in $[-1, 1]$, we can apply \cref{prop:bern-deriv-lb} to get that
\begin{align*}
    \deg(p) \geq \abs{p'(\xi)}\sqrt{1-\xi^2} \geq L\sqrt{1-\Delta^2}.
\end{align*}
We take the supremum over all $x, y$ to get the desired bound.
\end{proof}

We now discuss some implications of \cref{prop:generic-deriv-lb} for quantum algorithm design.
In recent works on quantum optimization \cite{vanApeldoornG19, BoulandGJSWT23}, approximations to $\exp(\beta y)$ on $y \in [-1, \delta]$ for constant $\delta$ are used to speed up Gibbs sampler subroutines for solving zero-sum games via QSVT.
It is a well-known result that a degree $O(\sqrt{\beta})$ polynomial approximates $\exp(-\beta y)$ up to additive error $0.1$ on $[-1, 0]$ \cite[Theorem 4.1]{sv14}.
However, only a $\approx \beta$ degree polynomial approximation was known when the polynomial is further required to be bounded in $[0, \delta]$ (a small interval outside the region of approximation).
Because the boundedness requirement comes from the use of QSVT, and is not needed in the classical setting, the state-of-the-art quantum runtime for zero-sum games~\cite{BoulandGJSWT23} incurs an overhead of $\sqrt{\beta} = \sqrt{1/\eps}$ compared to classical counterparts (while saving on dimension-dependent factors).
%In these applications, $\beta$ is taken to be roughly $\frac 1 \eps$ where $\eps$ is the desired optimization error.

\cref{cor:exp-deg-lb} applies \cref{prop:generic-deriv-lb} to show that for $\delta = \omega(\beta^{-1})$, adding the boundedness requirement necessitates an approximation of larger degree, up to quadratically worse when $\delta = \Omega(1)$.
This implies that the degree achieved by \cref{cor:exp_bounded} is nearly-tight.
This negatively resolves the open question posed by \cite{BoulandGJSWT23}, which was whether \cref{cor:exp-deg-lb} could be modified to remove the last remaining overhead in $\frac 1 \eps$ (when $\delta = \Omega(1)$).
We rule out this approach, suggesting it is necessary to fundamentally change the application of QSVT to obtain this conjectured speedup.

\begin{corollary}\label{cor:exp-deg-lb}
Let $\beta \ge 1$, $\delta \in (0, 1]$, and let $q(x)$ be a degree-$d$ polynomial which satisfies
\begin{align*}
\Abs{q(x) - \exp\Par{\beta x}} &\le 0.1 \text{ for } x \in [-1, 0] \text{ and } \\
\Abs{q(x)} &\le 1 \text{ for } x \in [0, \delta].
\end{align*}
Then $d = \Omega(\beta \sqrt{\delta})$.
\end{corollary}
\begin{proof}
Consider the change of variable $x = \frac{1 + \delta}{2}t - \frac{1 - \delta }{2}$ which maps $[-1, 1]$ to $[-1, \delta]$.
Then, for $f(t) = \exp(\beta x(t))$ and $p(t) = q(x(t))$, we know that $\abs{f(t) - p(t)} \leq 0.1$ for $t \in [-1, \tfrac{1-\delta}{1+\delta}]$ and $\abs{p(t)} \leq 1$ for $t \in [-1, 1]$.
Further,
\begin{align*}
    \frac{\abs{f(\tfrac{1-\delta}{1+\delta}) - f(\tfrac{2}{1+\delta}(-\tfrac{1}{\beta} + \tfrac{1-\delta}{2}))} - 0.2}{\abs{\tfrac{1-\delta}{1+\delta} - \tfrac{2}{1+\delta}(-\tfrac{1}{\beta} + \tfrac{1-\delta}{2})}}
    = \frac{1 - 1/e - 0.2}{\tfrac{2}{\beta(1+\delta)}} = \Omega(\beta).
\end{align*}
So, applying \cref{prop:generic-deriv-lb} with $S = \{\tfrac{2}{1+\delta}(-\tfrac{1}{\beta} + \tfrac{1-\delta}{2}),\tfrac{1-\delta}{1+\delta}\}$, we get that
\begin{equation*}
    d = \Omega\parens[\Big]{\beta\sqrt{1-\parens[\Big]{\frac{1-\delta}{1+\delta}}^2}} = \Omega(\beta\sqrt{\delta}). \qedhere
\end{equation*}
\end{proof}

A similar quadratic gap occurs for quantum algorithms for solving linear systems $\ma x = b$ when $\ma$ is positive definite~\cite{od21}.
Classical methods for this problem like conjugate gradient have a $\sqrt{\kappa}$ condition number dependence, which arises because $\tfrac 1 x$ has a good polynomial approximation on $[\tfrac 1 \kappa, 1]$ with degree $\approx \sqrt{\kappa}$.\footnote{
    We can see this explicitly.
    Approximating $x^{-1}$ on $[\kappa^{-1}, 1]$ is equivalent to approximating $\frac{1}{x-a}$ on $[-1, 1]$ for $a = 1 + \Theta(\kappa^{-1})$. There is an explicit expression for the Chebyshev coefficients of $\frac{1}{x-a} = \sum_{k=0}^\infty a_k T_k(x)$~\cite[Eq.\ (5.14)]{mh02}: $\abs{a_k} \sim \frac{1}{\sqrt{a^2 - 1}}(a - \sqrt{a^2 - 1})^k$.
    This is $\eps\norm{\frac{1}{x-a}}_{[-1, 1]}$ when taking $k = \Theta(\sqrt{\kappa})$.
}
However, QSVT requires approximations to be bounded on $[-1, 1]$; by applying a similar argument as in Corollary~\ref{cor:exp-deg-lb}, this implies that a degree of $\Omega(\kappa)$ is necessary to achieve same approximation quality with the boundedness constraint.
Orsucci and Dunjko work around this issue by observing that if we have a block-encoding of $\id - \ma$, then the function to be approximated is now $\tfrac 1 {1 - x}$, which can be done with degree $\approx \sqrt{\kappa}$, since the ill-conditioned part of the function is on the boundary of $[-1, 1]$, rather than the interior.

\subsection{Proof of Theorem~\ref{thm:main_bounded}}\label{ssec:bounded_proof}

We conclude with a proof of Theorem~\ref{thm:main_bounded}. Our proof builds upon several elementary bounds on Bernstein ellipses and the growth of Chebyshev polynomials, as well as the construction of explicit thresholding functions. For ease of exposition, we state all the helper bounds we use in this section, but defer their proofs to Appendix~\ref{app:more_csp}. We begin with our bounds on the sizes of Bernstein ellipses.

\begin{fact} \label{lem:bern-bounds}
    The Bernstein ellipse $E_\rho$ for $\rho \geq 1$ satisfies
    \[
        \operatorname{interior}(E_{\rho})
        \subset \Big\{x + \ii y \mid x, y \in \R,\, \abs{x} \leq \thalf(\rho + \rho^{-1}) \text{ and } \abs{y} \leq \thalf(\rho - \rho^{-1}) \Big\}.
    \]
    Further, for $\rho = 1 + \delta \leq 2$,
    \begin{align*}
        1 + \frac{\delta^2}{4} \leq \half(\rho + \rho^{-1})
        &= 1 + \frac{\delta^2}{2(1+\delta)} \leq 1 + \frac{\delta^2}{2}, \\
        \frac34 \delta \leq \half(\rho - \rho^{-1})
        &= \delta - \frac{\delta^2}{2(1+\delta)} \leq \delta.
    \end{align*}
\end{fact}

This yields the following containment fact, whose proof is deferred to Appendix~\ref{app:more_csp}.

\begin{restatable}{lemma}{restatetinyellipse} \label{lem:tiny-ellipse-rescale}
    For $\delta \in(0, 1)$, $(1+\delta)E_{1 + \alpha}$ is contained in the interior of $E_{\sigma}$, where $\sigma = 1 + 3(\alpha + \sqrt{\delta})$.
\end{restatable}

We also use the following bounds on Chebyshev polynomials, deferring a proof to Appendix~\ref{app:more_csp}. 

\begin{restatable}{lemma}{restatechebzbound}\label{lem:chebz-bound}
    There are universal constants $C, c > 0$ such that, for $n \geq 0$ and $x, y \in \mathbb{R}$, $|y| \le c$,
    \begin{equation*}
        \abs{T_n(x + \ii y)} \leq \begin{cases}
            (1 + C\sqrt{\abs{y}})^n & \abs{x} \leq 1 \\
            (x + \sqrt{x^2 - 1} + C\sqrt{\abs{xy}})^n & \abs{x} > 1
        \end{cases}.
    \end{equation*}
\end{restatable}

To ameliorate the polynomial growth of Chebyshev polynomials from Lemma~\ref{lem:chebz-bound}, we apply a threshold function with tails which decay superexponentially. Our thresholding is based on the Gaussian error function $\erf$; we define $\erf$ and recall some standard bounds on it in the following.

\begin{fact}[Eqs.\ 7.8.3, 7.8.7, \cite{DLMF}]\label{fact:erf}
    For $z \in \C$, $\erf: \C \to \C$ by
    $\erf(z) \defeq \frac{2}{\sqrt{\pi}} \int_0^{z} e^{-t^2}\dd t$. Then,
    \begin{align}
        1 - \erf(x) &= \frac{2}{\sqrt{\pi}}\int_x^\infty e^{-t^2} \dd t < \frac{2e^{-x^2}}{\sqrt{\pi}(1+x)} < 2e^{-x^2},
        & \label{eq:erf-bound} \\
        \abs{\erf(\ii x)} &= \frac{2}{\sqrt{\pi}}\int_0^{x} e^{t^2} \dd t < \frac{2(e^{x^2} - 1)}{\sqrt{\pi}\abs{x}} < 2e^{x^2} \text{ (when $x \geq 1$)}.
        \label{eq:erfi-bound}
    \end{align}
\end{fact}

For $z \in \R$, we note that $\half + \half \erf(z)$ is the cumulative distribution function for a Gaussian with mean $0$ and variance $\half$, which interpolates between $0$ and $1$; consequently, one may view $\erf$ (appropriately rescaled as necessary) as a ``smoothed'' variant of the sign function 
\begin{align*}
    \sgn(x) \defeq \begin{cases}
        -1 & x < 0 \\
        0 & x = 0 \\
        1 & x > 0
    \end{cases}.
\end{align*}
Building upon $\erf$, we state our family of thresholding functions, deferring proofs to Appendix~\ref{app:more_csp}. 

\begin{restatable}[Thresholding function]{lemma}{restaterectbounds}\label{lem:rect-bounds}
    For $\mu, s > 0$, let $r(z) \defeq \half(\erf(s(\mu + z)) + \erf(s(\mu - z)))$.
    When $z \in \R$, $0 \le r(z) \le 1$. When $x, y \in \R$ and $z = x + \ii y$, $|r(z) - r(x)| \le \exp(-s^2\Par{\mu - |x|}^2) |\erf(\ii sy)|$.
\end{restatable}

Evidently for $z \in \R$, the function $r$ behaves as a threshold: sufficiently inside $[-\mu, \mu]$, it is close to $1$, and sufficiently outside it is close to $0$. The size of the ``growth window'' near $\mu$ is roughly $\frac 1 s$, and Lemma~\ref{lem:rect-bounds} shows $r(x + \ii y) \approx r(x)$ for small $y$.
Leveraging these tools, we now prove Theorem~\ref{thm:main_bounded}.

\begin{proof}[Proof of Theorem~\ref{thm:main_bounded}]
Without loss of generality, we rescale so that $M = 1$.
To obtain the theorem statement, it  suffices to prove that there exists a polynomial $q$ of degree $O(\frac{b}{\delta}\log\frac{1}{\alpha\eps})$ such that
\begin{align}
    \norm{f - q}_{[-(1-\delta), 1-\delta]} &\leq \eps, \nonumber \\
    \norm{q}_{[-1, 1]} &\leq 1+\eps, \label{eq:bounded-true}\\
    \norm{q}_{[-b, -1] \cup [1,b]} &\leq \eps. \nonumber
\end{align}
To see this, consider $f$ as in the theorem statement.
Let $\delta' = \frac{\delta}{1+\delta} = \Theta(\delta)$, so that $\frac{1}{1-\delta'} = 1+\delta$.
Then $f(\tfrac{y}{1-\delta'})$ is analytic and bounded by $M$ for $y$ in the interior of $(1-\delta') E_\rho$, which contains $E_{1+\frac \alpha 4}$ by \cref{lem:tiny-ellipse-rescale} and  $\sqrt{\delta'} < \sqrt{\delta} < \frac \alpha C$.
Applying \eqref{eq:bounded-true}, we get a function $q(y)$ such that $q((1-\delta')x)$ satisfies the guarantees described above, with the intervals scaled up by a factor of $\frac{1}{1-\delta'}$:
\begin{align*}
    \abs{f(x) - q((1-\delta')x)} &\leq \eps
    & \text{for } x &\in [-1, 1], \\
    \abs{q((1-\delta')x)} &\leq 1+\eps
    & \text{for } x &\in [-\tfrac{1}{1-\delta'}, \tfrac{1}{1-\delta'}] = [-(1+\delta), 1+\delta], \\
    \abs{q((1-\delta')x)} &\leq \eps
    & \text{for } x &\in [-\tfrac{b}{1-\delta'}, -(1+\delta)] \cup [1+\delta,\tfrac{b}{1-\delta'}]
\end{align*}
To conclude, consider $\frac{1}{1+\eps}q((1-\delta')x)$.
We make $q$ slightly smaller so that it is bounded by $1$ in $[-(1+\delta), (1+\delta)]$.
This only affects the closeness to $f$ by a constant factor: for $x \in [-1,1]$,
\begin{align*}
    \abs{f(x) - \tfrac{1}{1+\eps}q((1-\delta')x)}
    \leq \abs{f(x) - q((1-\delta')x)} + (1 - \tfrac{1}{1+\eps})\abs{q((1-\delta')x)}
    \leq 2\eps.
\end{align*}
The degree of $\frac{1}{1+\eps}q((1-\delta')x)$ is degree of $q(x)$, as desired. We now proceed to prove \eqref{eq:bounded-true}.
By \cref{thm:trefethen}, there is a polynomial with degree $n = \lceil\frac 1 \alpha \log \frac{6}{\alpha\eps}\rceil$ with $\norm{f - f_n}_{[-1, 1]} \leq \frac \eps 3$, and the Chebyshev coefficients of $f_n = \sum_{k = 0}^n a_k T_k(x)$, satisfy $\abs{a_k} \leq 2\rho^{-k}$.
Next, let
$\tp(z) \defeq r(z)f_n(z)$
be the truncation $f_n$ multiplied by the function $r(z)$ from \cref{lem:rect-bounds} with 
\[
\mu \defeq 1 - \frac \delta 2,\; s \defeq \frac{C_s}{\delta}\sqrt{\log\frac{1}{\alpha\eps}},
\]
and $C_s$ is a constant to be chosen later. Let $\trho \defeq 1 + \frac \delta b$; we will show $\tp$ is bounded on $bE_{\trho}$, and then our final approximation $q$ will be an application of \cref{thm:trefethen} to approximate $\tp$ on $[-b, b]$. To this end, it suffices to bound $\tp(z)$ for all $z \in S$, where the strip $S$ is defined as
\[S \defeq \Brace{z = x + \ii y \mid |y| \le \delta},\]
because Fact~\ref{lem:bern-bounds} implies $S \supseteq bE_{\trho}$. We begin by bounding $r(z)$ for $z \in S$:
\begin{equation}\label{eq:rbound}
    \begin{aligned}
        \abs{r(x+\ii y)} &\leq r(x) + e^{-s^2(\mu - |x|)^2}\abs{\erf(\ii sy)}\\
        &\leq r(x) + e^{-s^2(\mu - |x|)^2}\abs[\Big]{\erf\parens[\Big]{\ii C_s\sqrt{\log\tfrac{1}{\alpha\eps}}}}\\
        &\leq r(x) + 2e^{-s^2(\mu - |x|)^2}(\alpha\eps)^{-C_s^2}.
    \end{aligned}
\end{equation}
The inequalities above respectively used \cref{lem:rect-bounds}, the definition of $S$, and \eqref{eq:erfi-bound}. We now combine \eqref{eq:rbound} with Lemma~\ref{lem:chebz-bound} to bound $\tp$ on $S$. First, consider when $z = x + \ii y \in S$ and $x \in [-1, 1]$.
This is the bottleneck of the argument, where $\tp$ is largest. We bound
\begin{align*}
    \abs{\tilde{p}(z)} = \abs{r(z)} \abs{f_n(z)} 
    &\le \parens[\Big]{1 + \exp\parens{-s^2(\mu - |x|)^2}\poly\parens[\Big]{\frac 1 {\alpha\eps}}} \abs[\Bigg]{\sum_{k = 0}^n a_k T_k(z)} \\
    &\le \poly\parens[\Big]{\frac 1 {\alpha\eps}}\parens[\Bigg]{2\sum_{k = 0}^n \rho^{-k}|T_k(z)|} \\
    &\le \poly\parens[\Big]{\frac 1 {\alpha\eps}} \parens[\Bigg]{\sum_{k = 0}^n \parens[\Big]{\frac{1 + K\sqrt{\delta}}{\rho}}^k} = \poly\parens[\Big]{\frac 1 {\alpha\eps}}.
\end{align*}
The first inequality was \eqref{eq:rbound}, the second used the guarantees of Theorem~\ref{thm:trefethen}, the third used Lemma~\ref{lem:chebz-bound}, and the last used $\delta \ll \alpha^2$, $n = \poly(\frac 1 {\alpha\eps})$.
Next, for $z = x + \ii y \in S$ with $|x| \ge 1$,
\begin{equation}\label{eq:bigx_rbound}
    \begin{aligned}
        r(x) = \half\Par{\erf\Par{s(\mu + |x|)} - \erf\Par{s(|x| - \mu)}} \le \half\Par{1 - \erf\Par{s(|x| - \mu)}} < e^{-s^2(|x| - \mu)^2},
    \end{aligned}
\end{equation}
where the first inequality was $\erf(z) \le 1$ for $z \in \R$, and the last was \eqref{eq:erf-bound}. Further, Lemma~\ref{lem:chebz-bound} yields
\begin{equation}\label{eq:bigx_fbound}
    \begin{aligned}
        |f_n(z)| &\le \sum_{k = 0}^n 2\rho^{-k}\Par{|x| + \sqrt{x^2 - 1} + K\sqrt{|xy|}}^k \\
        &\le 2n\Par{|x| + \sqrt{x^2 - 1} + K\sqrt{|xy|}}^n \le 2n\exp\Par{n\Par{|x| - 1 + \sqrt{x^2 - 1} + K\sqrt{|xy|}}}.
    \end{aligned}
\end{equation}
Continuing, we combine \eqref{eq:rbound}, \eqref{eq:bigx_rbound}, and \eqref{eq:bigx_fbound} to conclude
\begin{equation}\label{eq:tp_bigx_bound}
    \begin{aligned}
        |\tp(z)| &\le \exp\Par{-s^2(|x| - \mu)^2} \poly\parens[\Big]{\frac 1 {\alpha\eps}} |f_n(z)| \\
        &\le \exp\Par{-s^2(|x| - \mu)^2 + n\Par{|x| - 1 + \sqrt{x^2 - 1} + K\sqrt{|xy|}}} \poly\parens[\Big]{\frac 1 {\alpha\eps}} \\
        &\leq \exp\Par{\Par{-\frac{C_s^2}{\delta^2}(|x| - \mu)^2 + \frac 1{\sqrt{\delta}}\Par{|x| - 1 + \sqrt{x^2 - 1} + K\sqrt{|x|\delta}}  }\log \frac 1 {\alpha\eps}}\poly\Par{\frac 1 {\alpha\eps}}.
    \end{aligned}
\end{equation}
Here we used that $n \le \frac 1{\sqrt \delta} \log \frac 1 {\alpha\eps}$ under the assumed relationship between $\delta$ and $\alpha$, and the definition of $s$.
For sufficiently large $C_s$, it is straightforward to see that for all $|x| \ge 1$, since the left-hand side asymptotically grows faster than each term in the right-hand side,
\begin{equation}\label{eq:csdominates}
    \frac{C_s^2}{2\delta^2}\parens[\Big]{|x| - \parens[\Big]{1 - \frac \delta 2}}^2 \ge \frac 1 {\sqrt{\delta}}\parens[\Big]{|x| - 1 + \sqrt{x^2 - 1} + K\sqrt{|x|\delta}},
\end{equation}
and hence for this choice of $C_s$, plugging this into the previous bound gives that $\tp(x + \ii y) \le \poly(\frac 1 {\alpha\eps})$ for $|x| \ge 1$.\footnote{We note that the $\poly$ in \eqref{eq:tp_bigx_bound} hides a $C_s$-dependent exponent, which grows faster than \eqref{eq:csdominates}.}
Later, we will need a tighter bound when $y = 0$ and $|x| \ge 1$: in this setting, the second additive term in \eqref{eq:rbound} vanishes, and hence taking $C_s$ such that \eqref{eq:csdominates} holds, repeating the arguments in \eqref{eq:tp_bigx_bound} without the $\poly(\frac 1 {\alpha\eps})$ overhead gives for sufficiently large $C_s$,
\begin{equation}\label{eq:largex_tpbound}\tp(x) \le \exp\parens[\Big]{-\frac{C_s^2}{2\delta^2}(|x| - \mu)^2 \log \frac 1 {\alpha\eps}}\le \exp\parens[\Big]{-\frac{C_s^2}{8}\log\frac 1 {\alpha\eps}} \le \frac \eps 3 \text{ for all } x \in \R \text{ with } |x| \ge 1.\end{equation}
Thus, we have shown that for all $z \in S$, $|\tp(z)| \le \poly(\frac 1 {\alpha\eps})$, and as the product of analytic functions, $\tp$ is analytic. Next, for all $z \in \C$ let $\hp(z) \defeq \tp(bz)$. We have shown $\hp$ is bounded on $E_{\trho}$, and hence Theorem~\ref{thm:trefethen} gives a Chebyshev truncation $\hp_m$ such that $\norm{\hp - \hp_m}_{[-1, 1]} \le \frac \eps 3$, for
\[m = O\Par{\frac b \delta \log \frac b {\delta\eps}}.\]
Our final approximation is $q(z) \defeq \hp_m(\frac z b)$. By the definitions of $q$ and $\hp$, the relationship between $\hp$ and $\tp$ implies $\norm{q - \tp}_{[-b, b]} \le \frac \eps 3$. Combined with \eqref{eq:largex_tpbound}, this implies the third bound in \eqref{eq:bounded-true}, \[\norm{q}_{[-b, -1] \cup [1, b]} \le \eps.\]
The first bound $\norm{f - q}_{[-(1-\delta), 1-\delta]} \le \eps$ in \eqref{eq:bounded-true} follows from $\norm{q - \tp}_{[-(1-\delta), 1-\delta]} \le \frac \eps 3$, $\norm{f_n - f}_{[-(1-\delta), 1-\delta]} \le \frac \eps 3$, and $\norm{f_n - \tp}_{[-(1-\delta), 1-\delta]} \le (1 + \frac \eps 3)\norm{1 - r}_{[-(1-\delta), 1-\delta]} \le \frac \eps 3$ choosing $C_s$ sufficiently large. Finally, the second bound in \eqref{eq:bounded-true} follows from 
\begin{align*}
    \norm{q}_{[1-\delta, 1]} &\le \norm{q - \tp}_{[1-\delta, 1]} + \norm{\tp}_{[1-\delta, 1]} \le \frac \eps 3 + \norm{\tp}_{[1-\delta, 1]} \\
    &\le \frac \eps 3 + \norm{f_n}_{[1-\delta, 1]} \le \eps + \norm{f}_{[1-\delta, 1]} \le 1 + \eps,
\end{align*}
where we used the closeness bounds between $(\tp, q)$ and $(f_n, f)$, as well as the assumed bound on $f$ over $E_\rho$ (which contains $[1-\delta, 1]$). The bound $\norm{q}_{[-1, -(1-\delta)]} \le 1 + \eps$ follows symmetrically.
\end{proof}

In some of our applications in Section~\ref{sec:bounded}, we used the following property of our approximations constructed via Theorem~\ref{thm:main_bounded}, which we record here for convenience.

\begin{corollary}\label{cor:parity_bounded}
    In the setting of Theorem~\ref{thm:main_bounded}, if $f$ is even or odd, so is $q$.
\end{corollary}
\begin{proof}
    It is straightforward to check that all operations we perform on $f$ (Chebyshev truncation, multiplication by an even function $r$, and another Chebyshev truncation) preserve parity.
\end{proof}

\section*{Acknowledgements}

ET thanks t.f.\ for providing useful references.
ET also thanks David Gosset, Beni Yoshida, and Richard Cleve for the invitation to Waterloo and the hospitality; ET first read about the CS decomposition during a quiet night at the Perimeter Institute.
ET is supported by the NSF GRFP (DGE-1762114).
KT thanks Jonathan Kelner for encouraging him to learn about the CS decomposition, Christopher Musco for encouraging him to read \cite{trefethen19}, and Yang P.\ Liu for a helpful conversation about Jordan's lemma many moons ago.
We thank Carlo Beenakker for drawing our attention to the application of the CS decomposition to scattering theory.

\printbibliography
\appendix
\section{Proofs of quantum signal processing} \label{sec:qsp-proofs}

Our goal in this section is to characterize which polynomials are QSP-achievable.
Looking at the form of QSP, we can express its entries via polynomials satisfying a recurrence.
\begin{lemma}[QSP as a recurrence] \label{lem:qsp-recurrence}
    For some phase factors $\Phi = \{\phi_j\}_{0 \leq j \leq n} \in \mathbb{R}^{n + 1}$,
    \begin{align}
        \qsp(\{\phi_j\}_{k \leq j \leq n}, x) = \begin{pmatrix}
        p_k(x) & q_k^*(-x)\sqrt{1-x^2} \\
        q_k(x)\sqrt{1-x^2} & p_k^*(-x)
    \end{pmatrix},
    \end{align}
    where $p_n(x) = e^{\ii \phi_n}$ and $q_n(x) = 0$, and $p_k(x)$ and $q_k(x)$ satisfy the following recurrence relation:
    \begin{align}
        p_{k}(x) &= e^{\ii \phi_{k}}(xp_{k+1}(x) + (1-x^2)q_{k+1}(x)), \\
        q_{k}(x) &= e^{-\ii \phi_{k}}(p_{k+1}(x) - xq_{k+1}(x)).
    \end{align}
\end{lemma}
\begin{proof}
The base case is because
\begin{align*}
    \qsp(\{\phi_n\}, x) = \begin{pmatrix}
        e^{\ii \phi_n} & 0 \\ 0 & e^{-\ii \phi_n}
    \end{pmatrix}.
\end{align*}
The inductive case is a computation:
\begin{align*}
    & \qsp(\{\phi_j\}_{k \leq j \leq n}, x) \\
    &= e^{\ii \phi_{k} \msigma_z} \mr(x) \cdot \qsp(\{\phi_j\}_{k+1 \leq j \leq n}, x) \\
    &= \begin{pmatrix}
        e^{\ii \phi_{k}} x & e^{\ii \phi_{k}} \sqrt{1-x^2} \\
        e^{-\ii \phi_{k}} \sqrt{1-x^2} & -e^{-\ii \phi_{k}} x
    \end{pmatrix}\begin{pmatrix}
        p_{k+1}(x) & q_{k+1}^*(-x)\sqrt{1-x^2} \\
        q_{k+1}(x)\sqrt{1-x^2} & p_{k+1}^*(-x)
    \end{pmatrix} \\
    &= \begin{pmatrix}
        e^{\ii \phi_{k}}(x p_{k+1}(x) + (1-x^2)q_{k+1}(x)) & e^{\ii \phi_{k}}(p_{k+1}^*(-x) + x q_{k+1}^*(-x))\sqrt{1-x^2} \\
        e^{-\ii \phi_{k}}(p_{k+1}(x) - x q_{k+1}(x))\sqrt{1-x^2} & e^{-\ii \phi_{k}}(-x p_{k+1}^*(-x) + (1-x^2)q_{k+1}^*(-x))
    \end{pmatrix} \\
    &= \begin{pmatrix}
        p_{k}(x) & q_{k}^*(-x)\sqrt{1-x^2} \\
        q_{k}(x)\sqrt{1-x^2} & p_{k}^*(-x)
    \end{pmatrix}. \qedhere
\end{align*}
\end{proof}

\begin{theorem}[{Variant of \cite[Theorem 3]{gslw18}}] \label{thm:qsp-characterization}
    A degree-$n$ polynomial $p(x) \in \mathbb{C}[x]$ is QSP-achievable if and only if there is a polynomial $q(x)$ such that:
    \begin{enumerate}[label=\textup{(\alph*)}]
        \item $q$ has degree $\leq n-1$;
        \item $(p, q)$ are (even, odd) or (odd, even);
        \item $\abs{p(x)}^2 + (1-x^2)\abs{q(x)}^2 \equiv 1$.
    \end{enumerate}
\end{theorem}
\begin{proof}
First, we consider the ``only if'' direction.
Suppose $p(x)$ is QSP-achievable with the phase factors $\Phi \in \mathbb{R}^{n+1}$.
Then, by \cref{lem:qsp-recurrence}, there is some $q(x)$ such that
\begin{align*}
    \qsp(\Phi, x) = \begin{pmatrix}
        p(x) & q^*(-x)\sqrt{1-x^2} \\
        q(x)\sqrt{1-x^2} & p^*(-x)
    \end{pmatrix},
\end{align*}
derived from the recurrence described in that lemma.
From this recurrence, we can verify that at all times, conditions (a) and (b) are satisfied.
Finally, condition (c) is always satisfied because $\qsp(\Phi, x)$ is a product of unitary matrices, and so is unitary: the first column having norm one is equivalent to $\abs{p(x)}^2 + (1-x^2)\abs{q(x)}^2 = p(x)p^*(x) + (1-x^2)q(x)q^*(x) = 1$, and this argument works for every $x \in [-1, 1]$.
Because it holds for infinitely many $x$, the equality holds as polynomials.

Second, we consider the ``if'' direction.
Suppose we have some $p(x)$ of degree $n$ and $q(x)$ satisfying (a), (b), and (c).
We want to construct phase factors that implement $p(x)$.
We proceed by induction: when $n = 0$, this means that $p(x)$ is scalar and $q(x)$ has degree $\leq -1$ (meaning it must be zero).
Thus, $p(x) \equiv e^{\ii \phi}$ for some $\phi$; we can implement this with $\Phi = \{\phi\}$.
For the inductive step, consider $p(x)$ of degree $n+1$.
If we could show that there exists some $\varphi$ such that 
\begin{align} \label{eq:qsp-characterization-if}
    (e^{\ii \varphi \msigma_z} \mr(x))^\dagger \begin{pmatrix}
        p(x) & q^*(-x)\sqrt{1-x^2} \\
        q(x)\sqrt{1-x^2} & p^*(-x)
    \end{pmatrix}
    &= \begin{pmatrix}
        p_\downarrow(x) & q_\downarrow^*(-x)\sqrt{1-x^2} \\
        q_\downarrow(x)\sqrt{1-x^2} & p_\downarrow^*(-x)
    \end{pmatrix}
\intertext{
for $p_\downarrow, q_\downarrow$ some even/odd polynomials of one degree lower than $p$ and $q$, then we would be done.
By assumption, the matrices on the left-hand side of \eqref{eq:qsp-characterization-if} are unitary, so the right-hand side matrix is also unitary.
Thus, $p_\downarrow$ and $q_\downarrow$ satisfy all the properties of the induction hypothesis, and there are phase factors $\{\phi_0,\ldots,\phi_n\} \in \mathbb{R}^{n+1}$ giving the equality}
    (e^{\ii \varphi \msigma_z} \mr(x))^\dagger \begin{pmatrix}
        p(x) & q^*(-x)\sqrt{1-x^2} \\
        q(x)\sqrt{1-x^2} & p^*(-x)
    \end{pmatrix}
    &= \qsp(\{\phi_0, \ldots, \phi_n\}, x). \\
    \begin{pmatrix}
        p(x) & q^*(-x)\sqrt{1-x^2} \\
        q(x)\sqrt{1-x^2} & p^*(-x)
    \end{pmatrix}
    &= \qsp(\{\varphi, \phi_0, \ldots, \phi_n\}, x)
\end{align}
So it comes down to finding the right value of $\varphi$ that could remove a degree from $p$ and $q$ in \eqref{eq:qsp-characterization-if}.
By properties (a) and (b), we can write
\begin{align}
    p(x) &= a_{n+1} x^{n+1} + a_{n-1}x^{n-1} + \ldots \\
    q(x) &= b_{n} x^{n} + a_{n-2}x^{n-2} + \ldots
\end{align}
The condition (c) implies that $\abs{a_{n+1}} = \abs{b_n}$.
Now we perform the matrix calculation.
Since $\mr(x)$ is its own inverse, $(e^{\ii \varphi \msigma_z} \mr(x))^\dagger = \mr(x) e^{-\ii \varphi \msigma_z}$, so
\begin{align}
    & \quad (e^{\ii \varphi \msigma_z} \mr(x))^\dagger \begin{pmatrix}
        p(x) & q^*(-x)\sqrt{1-x^2} \\
        q(x)\sqrt{1-x^2} & p^*(-x)
    \end{pmatrix} \\
    &= \begin{pmatrix}
        e^{-\ii \varphi} x & e^{\ii \varphi} \sqrt{1-x^2} \\
        e^{-\ii \varphi} \sqrt{1-x^2} & -e^{\ii \varphi} x
    \end{pmatrix}\begin{pmatrix}
        p(x) & q^*(-x)\sqrt{1-x^2} \\
        q(x)\sqrt{1-x^2} & p^*(-x)
    \end{pmatrix} \\
    &= \begin{pmatrix}
        e^{-\ii \varphi} p(x) + e^{\ii \varphi}(1-x^2)q(x) & (e^{\ii \varphi} p^*(-x) + e^{-\ii \varphi}xq^*(-x))\sqrt{1-x^2} \\
        (e^{-\ii \varphi}p(x) - e^{\ii \varphi} xq(x))\sqrt{1-x^2} & -e^{\ii \varphi}xp^*(-x) + e^{-\ii \varphi}(1-x^2)q^*(-x)
    \end{pmatrix}
\end{align}
So, we need the following polynomials to have lower degree:
\begin{align}
    p_\downarrow(x) &= e^{-\ii \varphi}p(x) + e^{\ii \varphi}(1-x^2)q(x) \\
    q_\downarrow(x) &= e^{-\ii \varphi}p(x) - e^{\ii \varphi}xq(x)
\end{align}
The ``leading'' coefficient of $x^{n+1}$ for $p_\downarrow$ and $x^n$ for $q_\downarrow$ are the same: $e^{-\ii \varphi} a_{n+1} - e^{\ii \varphi} b_n$.
If we choose $\varphi$ such that $e^{\ii \varphi} = \sqrt{a_{n+1}/b_n}$, then this coefficient is 0, and so the degrees of $p_\downarrow$ and $q_\downarrow$ are $\leq n-1$ and $\leq n-2$, as desired.
\end{proof}

The characterization of when a polynomial $p(x)$ is QSP-achievable is still somewhat difficult to understand.
With more work, we can give a clearer understanding of QSP-achievable polynomials, if we give up the imaginary degree of freedom in our polynomials.
Generalizing the notion of $p$ being QSP-achievable, we say that the pair of polynomials $(p, q)$ is QSP-achievable if there are phase factors such that $p$ and $q$ are the two polynomials in the characterization of \cref{lem:qsp-recurrence}.
\begin{theorem}[{\cite[Theorem 5, Lemma 6]{gslw18}}] \label{thm:real-characterization}
    Let $p_{\Re}(x), q_{\Re}(x) \in \R[x]$ be real-valued polynomials with $p$ of degree $n$.
    Then there exist $p, q \in \C[x]$ such that $(p, q)$ is QSP-achievable and $p_{\Re} = \Re(p)$, $q_{\Re} = \Re(q)$ if and only if
    \begin{enumerate}[label=\textup{(\alph*)}]
        \item $q_{\Re}$ has degree $\leq n-1$;
        \item $(p_{\Re}, q_{\Re})$ are (even, odd) or (odd, even);
        \item[\textup{(c')}] $(p_{\Re}(x))^2 + (1-x^2)(q_{\Re}(x))^2 \leq 1$ for $x \in [-1,1]$.
    \end{enumerate}
\end{theorem}
To interpret this claim, it implies that if we have real polynomials where the ``unit norm'' constraint is merely an inequality (c'), then we can add imaginary components to make it an equality, so that by \cref{thm:qsp-characterization} these supplemented polynomials are achievable.
\cref{thm:real-one} follows as a corollary of this theorem, taking $q = 0$.
\begin{proof}
The ``only if'' direction is more straightforward: if $(p, q)$ is QSP-achievable, then the real parts of $p$ and $q$ satisfy (a), (b), and (c') by \cref{thm:qsp-characterization}.

The ``if'' direction requires some work: given $p_{\Re}$ and $q_{\Re}$, we need to find some $p_{\Im} \in \R[x]$ and $q_{\Im} \in \R[x]$ of the right degree and parity such that $p \defeq p_{\Re} + \ii p_{\Im}$ and $q \defeq q_{\Re} + \ii q_{\Im}$ satisfy
\begin{align*}
    \abs{p(x)}^2 + (1-x^2)\abs{q(x)}^2 = p_{\Re}^2 + p_{\Im}^2 + (1-x^2)(q_{\Re}^2 + q_{\Im}^2) \equiv 1.
\end{align*}
This would imply the conclusion via \cref{thm:qsp-characterization}.
Consider $P = 1 - p_{\Re}^2 - (1-x^2)q_{\Re}^2$, which is an even polynomial with real coefficients.
By assumption (c'), we know $P$ is non-negative in $x \in [-1,1]$.
If we can write $P = A^2 + (1-x^2)B^2$ where $\deg(A) \leq \deg(P)$, $\deg(B) \leq \deg(P) - 1$, and $(A, B)$ are (odd, even) or (even, odd), then we are done.
If we have two polynomials $P$ and $Q$ that can be expressed in the above way, then their product can:
\begin{align*}
    PQ &= (A_P^2 + (1-x^2)B_P^2)(A_Q^2 + (1-x^2)B_Q^2) \\
    &= \abs{A_P+ \ii\sqrt{1-x^2}B_P}^2\abs{A_Q + \ii\sqrt{1-x^2}B_Q}^2 \\
    &= \abs{(A_PA_Q - (1-x^2)B_PB_Q) + \ii\sqrt{1-x^2}(A_P B_Q + A_Q B_P)}^2 \\
    &= (A_PA_Q - (1-x^2)B_PB_Q)^2 + (1-x^2)(A_P B_Q + A_Q B_P)^2.
\end{align*}
So, it suffices to prove that this representation is possible for all ``irreducible'' polynomials, i.e.\ even polynomials with real coefficents that are non-negative in $[-1, 1]$, that cannot be decomposed into the product of two polynomials satisfying the same criteria.
Using the fundamental theorem of algebra, we can give a complete list of such polynomials up to scaling by a positive number.
\begin{enumerate}
    \item (Polynomials with roots $(0, 0)$) $R(x) = x^2$; here, $A = x$ and $B = 0$.
    \item (Polynomials with roots $(-s, -s, s, s)$ for $s \in (0, 1)$) $R(x) = (x^2-s^2)^2$; here, $A = (x^2-s^2)$ and $B = 0$.
    \item (Polynomials with roots $(-s, s)$ for $s \geq 1$) $R(x) = s^2 - x^2$; here, $A = \sqrt{s^2 - 1}x$ and $B = s$.
    \item (Polynomials with roots $(-\ii s, \ii s)$ for $s > 0$) $R(x) = x^2 + s^2$ for $C > 0$; here, $A = \sqrt{s^2 + 1}x$ and $B = s$.
    \item (Polynomials with roots $(s + \ii t, s - \ii t, -s + \ii t, -s - \ii t)$ for $s, t > 0$) 
    $R(x) = x^4 + 2x^2(t^2 - s^2) + (s^2 + t^2)^2$; here, $A = cx^2 - (s^2 + t^2)$ and $B = \sqrt{c^2 - 1}x$ for $c = s^2 + t^2 + \sqrt{2(s^2 + 1)t^2 + (s^2 - 1)^2 + t^4}$.
\end{enumerate}
Because all of these polynomials can be written in the desired representation, all polynomials satisfying the criteria can.
\end{proof}
\section{More applications of the CS decomposition}\label{app:csd_interpret}

To shed light on the CS decomposition as capturing interactions between subspaces, in this section we derive two further applications beyond QSVT. We note that we do not claim the intuition provided in this section is very helpful for understanding the particular application of QSVT. However, we hope the reader is sufficiently convinced of the virtue of the CS decomposition as a technical tool, and this section serves to provide additional background on this tool.

\subsection{Principal angles}

In this section we consider two rank-$a$ subspaces $\xset = \Image(\mpi_x) \subset \C^d$ and $\yset = \Image(\mpi_y) \subset \C^d$, for some $a \in [d]$.
For $k \in [a]$, we define the $k^{\text{th}}$ \emph{principal angle} between $\xset$ and $\yset$ recursively via
\begin{equation}\label{eq:svd_big}\cos(\theta_k) \defeq \max_{\substack{x \in \xset \\ \norm{x}_2 = 1}} \max_{\substack{y \in \yset \\ \norm{y}_2 = 1}} \inprod{x}{y} \text{ subject to } x \perp x_i, y \perp y_i \text{ for all } i < k, \end{equation}
where $x_k$, $y_k$ are the \emph{principal vectors} realizing the maximum above. In other words, the first principal angle $\theta_1$ is the largest angle between a vector in $\xset$ and a vector in $\yset$; $\theta_2$ is the largest angle between vectors in the subspaces of $\xset$ and $\yset$ orthogonal to the vectors achieving $\theta_1$, and so on. This definition only depends on the subspaces, and so is agnostic to the choice of basis for $\xset$ and $\yset$.

\begin{lemma} \label{lem:svs-are-angles}
Let $\mx, \my \in \C^{d \times a}$ be such that their columns are orthonormal bases for $\xset$ and $\yset$, respectively.
Then, the values $\{\cos(\theta_k)\}_{k \in [a]}$ are the singular values of $\mx^\dagger \my$.
Further, letting $\mv \mc \mw^\dagger$ be the SVD of $\mx^\dagger \my$, the principal vectors between $\xset, \yset$ are columns of $\mx\mv$, $\my\mw$.
\end{lemma}
\begin{proof}
This result follows from the variational characterization of singular values and vectors.
Let $\mc = \diag{c}$. Fix $k \in [a]$ and suppose inductively the conclusion holds for all $i < k$. Recall that the SVD is recursively defined by
\begin{equation}\label{eq:svd_small}c_k = \max_{\norm{v}_2 = \norm{w}_2 = 1} v^\dagger \mx^\dagger \my w, \text{ subject to } v \perp v_i,\; w \perp w_i \text{ for all } i < k. \end{equation}
Inductively assume $x_i = \mx v_i$ for all $i < k$. Notice that every unit vector in $\xset$ can be written as $\mx v$ for some unit $v \in \C^a$, and since $\mx^\dagger \mx = \id$, we have $\mx v \perp \mx v_i$ iff $v \perp v_i$. By reasoning similarly for $\yset$, we conclude the optimization problems \eqref{eq:svd_big} and \eqref{eq:svd_small} are the same under the transformation $x \gets \mx v$ and $y \gets \my w$. Hence setting $x_k \gets \mx v_k$ and $y_k \gets \my w_k$ we may continue inducting.
\end{proof}

Now we explain the connection between the above digression and the CS decomposition.
\cref{lem:svs-are-angles} shows that we can find the principal angles between $\xset$ and $\yset$ by taking orthogonal bases, $\mx$ and $\my$, and computing the SVD of $\mx^\dagger \my$.
We could further ask: what are the principal angles of the orthogonal subspaces, $\xset_\perp = \{u \mid \langle u, x\rangle = 0 \text{ for all } x \in \xset\}$ and $\yset_\perp = \{u \mid \langle u, y\rangle = 0 \text{ for all } y \in \yset\}$?
We can apply the same lemma on bases for the subspaces, $\mx^\perp$ and $\my^{\perp}$, to compute them; however, we can say more. First, let the following matrices be unitary completions of $\mx$, $\my$:
\begin{equation}\label{eq:complete_unitary}
\begin{pmatrix}
\mx & \mx_\perp
\end{pmatrix},\; \begin{pmatrix}
\my & \my_\perp
\end{pmatrix}.
\end{equation}
We next take the CS decomposition (Theorem~\ref{thm:cs}) of the product (which is also unitary),
\[\mmu = \begin{pmatrix}
\mx & \mx_\perp
\end{pmatrix}^\dagger \begin{pmatrix}
\my & \my_\perp
\end{pmatrix} = \begin{pmatrix} \mx^\dagger \my & \mx^\dagger \my_\perp \\ \mx_\perp^\dagger \my & \mx_\perp^\dagger \my_\perp\end{pmatrix}.\]
This gives us $\mv_1, \mv_2, \mw_1, \mw_2$ such that
\[\begin{pmatrix}\mv_1^\dagger \mx^\dagger \my \mw_1 & \mv_1^\dagger \mx^\dagger \my_\perp \mw_2 \\ \mv_2^\dagger \mx_\perp^\dagger \my \mw_1 & \mv_2^\dagger \mx_\perp^\dagger \my_\perp \mw_2 \end{pmatrix} = \begin{pmatrix} \md_{11} & \md_{12} \\ \md_{21} & \md_{22} \end{pmatrix},\]
of the form in Theorem~\ref{thm:cs}.
This gives simultaneous SVDs for each block, and thus gives the principal angles and vectors for all combinations of $\xset, \xset_\perp, \yset$, and $\yset_\perp$.
In particular, we can see that the principal angles of $(\xset, \yset)$ and $(\xset_\perp, \yset_\perp)$ are related: up to padding by 0's and 1's, they are identical!

Further, we can take $\mx \gets \mx \mv_1$, $\my \gets \my \mw_1$, etc.\ without affecting the induced subspaces (since e.g.\ $\mx \mv_1 \mv_1^\dagger \mx^\dagger = \mx \mx^\dagger$), but such that after this transformation we simply have
\begin{equation}\label{eq:cs_canonical}\begin{pmatrix}\mx^\dagger \my & \mx^\dagger \my_\perp \\ \mx_\perp^\dagger \my & \mx_\perp^\dagger \my_\perp \end{pmatrix} = \begin{pmatrix} \md_{11} & \md_{12} \\ \md_{21} & \md_{22} \end{pmatrix}.\end{equation}
That is, we can choose canonical basis representations of $\xset$, $\yset$, and their complement subspaces consistently, such that every pairing directly induces the ``principal angles and vectors'' defined above. We remark that this argument extends just fine to $\xset$, $\yset$ of different dimensions.

\subsection{Jordan's lemma}

Next we derive Jordan's lemma \cite{jordan75}, a useful way of decomposing $\C^d$ into subspaces (induced by a unitary matrix) which are jointly compatible with two projection matrices in a certain sense. We note that Jordan's lemma has seen varied implicit or explicit uses in the quantum computing literature (including a suggestion by \cite{gslw18}), and refer to \cite{reg06} for an account of this.

\begin{lemma}\label{lem:jordan}
Let $\mpi_x, \mpi_y \in \C^{d \times d}$ be projection matrices. There exists a unitary matrix $\mmu \in \C^{d \times d}$ with columns $\{u_i\}_{i \in [d]}$, and a partition of $[d]$ into $\calS \defeq \{S_j\}_{j \in [k]}$ such that $|S_j| \in \{1, 2\}$ for all $j \in [k]$, and $\mmu^\dagger \mpi_x \mmu$, $\mmu^\dagger \mpi_y \mmu$ are block-diagonal with blocks indexed by $\{S_j\}_{j \in [k]}$; each block is trace-$1$. Moreover for each $S_j = \{i, i'\}$ where $|S_j| = 2$, we have $\mpi_x u_i \parallel \mpi_x u_{i'}$ and $\mpi_y u_i \parallel \mpi_y u_{i'}$.
\end{lemma}

In other words, there is a choice of subspaces given by $\mmu$ (whose columns are partitioned by $\calS$) such that if $i, i' \in [n]$ where $i \in S_j$ and $i' \in S_{j'}$, $u_i^\dagger \mpi_x u_{i'} \neq 0$, $u_i^\dagger \mpi_y u_{i'} \neq 0$ iff $j = j'$. Moreover, the second part states each of the $2 \times 2$ blocks in $\mmu^\dagger \mpi_x \mmu$ and $\mmu^\dagger \mpi_y \mmu$ are in fact rank-$1$ and trace-$1$.

\begin{proof}[Proof of Lemma~\ref{lem:jordan}]
We first prove this in the special case when $\mpi_x$ and $\mpi_y$ are dimension-$\frac d 2$ projectors with ``no intersection,'' and briefly discuss how to extend this to the general case. \\

\textit{Special case.} Suppose $\mpi_x$ and $\mpi_y$ are dimension-$\frac d 2$ projectors, and let us further make one following restriction. Let $\mpi_x = \mx\mx^\dagger$ and $\mpi_y = \my\my^\dagger$ where $\mx$, $\my$ and their ``completions'' $\mx_\perp$, $\my_\perp$ (in the sense that the matrices \eqref{eq:complete_unitary} are unitary) are chosen such that \eqref{eq:cs_canonical} holds, as guaranteed by Theorem~\ref{thm:cs}; we make the restriction that for $\mc = \mx^\dagger \my$, all of the diagonal entries of $\mc$ are in $(0, 1)$. The previous section shows this means $\Span(\mx) \cap \Span(\my) = \emptyset$. We choose the basis inducing $\mmu$ as follows. Let the columns of $\mx$ be $\{x_j\}_{j \in [\frac d 2]} \subset \C^d$ and the columns of $\my$ be $\{y_j\}_{j \in [\frac d 2]} \subset \C^d$. For $j \in [\frac d 2]$, we let $\{u_{2j - 1}, u_{2j}\} \subset \C^d$ be an arbitrary basis of $\Span\{x_j, y_j\}$. We claim such $\mmu$ meets the requirements, where each $S_j = \{2j - 1, 2j\}$ in our partition. For all $j \in [\frac d 2]$, since $\mx^\dagger \my = \mc$,
\[\mproj_x y_j = \mx\mx^\dagger y_j = \mx \Par{\mc e_j} = c_jx_j. \]
So, $\mproj_x$ maps $\Span\{x_j, y_j\}$ to $\Span\{x_j\}$, and similarly $\mproj_y$ maps $\Span\{x_j, y_j\}$ to $\Span\{y_j\}$. This proves the second part of the lemma, namely that $\mproj_x$ and $\mproj_y$ act as rank-$1$ projectors on each block of the partition. For the first part (the block-diagonal structure), it suffices to show that for all $j \neq j' \in [\frac d 2]$, $x_j \perp \Span\{x_{j'}, y_{j'}\}$, since we already argued $\mproj_x$ maps $\Span\{x_j, y_j\}$ to $\Span\{x_j\}$. By orthonormality of $\mx$, $x_j \perp x_{j'}$, and since we used the canonical choice where $\mx^\dagger \my = \mc$, indeed $x_j \perp y_{j'}$ as well. To see that each block has trace $1$, write $x_j = \alpha_j u_{2j - 1} + \beta_j u_{2j}$. We have that $\mx\mx^\dagger u_{2j - 1} = \alpha_j x_j$ and $\mx\mx^\dagger u_{2j} = \beta_j x_j$; to see this, we already argued that $u_{2j - 1}$, $u_{2j}$ are orthogonal to all $x_i$ for $i \neq j$, since $x_j \perp x_i$ and $y_j \perp x_i$. Hence, the $2 \times 2$ block of $\mproj_x$ indexed by $S_j$ is the outer product of $\begin{pmatrix} \alpha_j & \beta_j \end{pmatrix}$ which clearly has trace $1$; a similar argument applies to $\mproj_y$.\\

\textit{General case.} More generally, we can ``pull out'' vectors corresponding to the $\id$ and $\mzero$ blocks of the decomposition in Theorem~\ref{thm:cs}, when the dimensions are unequal. Concretely, again let $\mproj_x = \mx \mx^\dagger$ and $\mproj_y = \my\my^\dagger$, such that $\mx^\dagger \my = \mc$ and $\mc$ has the form guaranteed by Theorem~\ref{thm:cs}. Further, assume $d_1 = \text{dim}(\Span(\mx)) \ge \text{dim}(\Span(\my)) = d_2$. Whenever there is a $1$ entry in $\mc$, this corresponds to a subset of size $1$ in the partition with the column of $\mmu$ set to the corresponding vector in $\textup{Span}(\mx) \cap \textup{Span}(\my)$. Whenever there is a $0$ entry we simply pull out the corresponding vector in $\textup{Span}(\mx) \setminus \textup{Span}(\my)$ into its own block in the partition. Finally, when $\textup{Span}(\mx) \oplus \textup{Span}(\my) \neq \C^d$ we find any orthonormal basis of $(\textup{Span}(\mx) \oplus \textup{Span}(\my))^c$ and add them it as columns of $\mmu$. It is an exercise to check the overall dimension of $\mmu$ after this process is $d$.

\end{proof}
\section{Proof of Carlini's formula}\label{app:carlini}

\begin{proof}[Proof of \cref{carlini}]
Recall $2I_n(t)$ is the $n^{\text{th}}$ Chebyshev coefficient for $\exp(tx)$. It will be slightly more convenient for us to reparameterize and bound $I_n(nt)$.
Without loss of generality $t \geq 0$, since $\exp(tx)$ and $\exp(-tx)$ have the same Chebyshev coefficients up to sign.
By \eqref{eq:cheb-integral}, for $n \geq 1$,
\begin{align*}
    I_{n}(nt) &= \frac{1}{2\pi\ii} \oint_{\abs{z} = 1} z^{-n-1} \exp(\tfrac{nt}{2}(z + z^{-1})) \dd z \\
    &= \frac{1}{2\pi \ii}\oint_{\abs{z} = 1} \exp\Big(-n\Big(\log(z)-\frac{t}{2}(z + \tfrac1z)\Big)\Big) \frac{\dd z} z.
\end{align*}
We choose a contour circling around the origin once; by Cauchy's theorem, this results in the same integral as the contour does not cross the origin. We parameterize it via $z = re^{\ii \theta}$ and construct $r$ as a function of $\theta$.
Consider the (rescaled) imaginary part of the expression in the exponential:
\begin{align*}
    \log(z) - \frac{t}{2}(z + \tfrac1z)
    &= \log(r) + \ii \theta - \frac{t}{2n}(r e^{\ii \theta} + \tfrac1 r e^{-\ii \theta}) \\
   \implies \Im\Big(\log(z) -\frac{t}{2}(z + \tfrac1z)\Big)
    &= \theta -\frac{t}{2}(r\sin(\theta) - \tfrac1r \sin(\theta)). \\
    \intertext{
        We wish to make the imaginary part constant; we set it equal to $\psi$, and solve for $r$:
    }
    \psi &= \theta - \frac{t}{2}(r\sin(\theta) - \tfrac1r \sin(\theta)) \\
    \implies \thalf(r - \tfrac1r) &= \frac{\theta - \psi}{t\sin(\theta)} \\
    \implies r &= \frac{\theta - \psi}{t\sin(\theta)} + \sqrt{\Par{\frac{\theta - \psi}{t\sin(\theta)}}^2 + 1}.
    \end{align*}
        Above, we use that $r = x + \sqrt{x^2 + 1}$ and $\frac1r = -x + \sqrt{x^2 + 1}$ is a solution to $\half(r - \frac1r) = x$.
        Now, we choose to take our contour for $\theta$ from $-\pi+\psi$ to $\pi+\psi$, and we take $\psi = 0$.
        So the contour is
    \begin{align*}
    z = re^{\ii\theta} = \Par{\frac{\theta}{t\sin(\theta)} + \sqrt{\Par{\frac{\theta}{t\sin(\theta)}}^2 + 1}}e^{\ii \theta} \text{ for } \theta \in [-\pi, \pi] \text{, where } \frac {\sin 0} 0 \defeq 1.
    \end{align*}
Since $r \geq 1$ always, the contour winds once counter-clockwise around zero and is valid for evaluating the integral.
        By design, on this contour $\Im(\log(z) - \frac{t}{2}(z - \tfrac1z))$ vanishes, and the real part is
    \begin{align*}
    F(\theta, t) &\defeq \Re\Big(\log(z) - \frac{t}{2}(z + \tfrac1z)\Big)
    = \log(r) - \frac{t}{2}(r + \tfrac1r)\cos(\theta) \\
    &= \log\Par{\frac{\theta}{t\sin(\theta)} + \sqrt{\Par{\frac{\theta}{t\sin(\theta)}}^2 + 1}} - t\cos(\theta)\sqrt{\Par{\frac{\theta}{t\sin(\theta)}}^2 + 1}.
\end{align*}
Now, we consider the original integral along this contour:
\begin{align*}
    I_{n}(nt) &= \frac{1}{2\pi \ii}\oint \exp\Big(-n\Big(\log(z) - \frac{t}{2}(z + \tfrac1z)\Big)\Big) \frac{\dd z} z\\
    &= \frac{1}{2\pi \ii}\int_{-\pi}^\pi \exp\Big(-n \cdot F(\theta, t)\Big) \Big(\frac{\dd r(\theta)}{\dd\theta} \cdot \frac 1 {r(\theta)}+ \ii\Big) \dd\theta \\
    &= \frac{1}{2\pi}\int_{-\pi}^{\pi} \exp\Big(-n \cdot F(\theta, t)\Big) \dd\theta \\
    &= \frac{1}{\pi}\int_{0}^{\pi} \exp\Big(-n \cdot F(\theta, t)\Big) \dd\theta.
\end{align*}
The last two lines use that as $F(\theta, t)$ and $r(\theta)$ are even functions in $\theta$, the piece of the integral corresponding to $\frac{\dd r(\theta)}{r(\theta)\dd\theta}$ vanishes.
From here, it becomes a matter of bounding $F(\theta, t)$.
%The simple bound $F(\theta, t) \geq F(0, t) = \log(n/t + \sqrt{(n/t)^2 + 1}) - \sqrt{(t/n)^2 + 1}$ gives Kapteyn's inequality used above, which comforts us that the above derivation is correct.
We compute
% https://www.wolframalpha.com/input?i=simplify+derivative+of+log%28x%2F%28C+sin%28x%29%29+%2B+sqrt%28x%5E2%2F%28C+sin%28x%29%29%5E2+%2B+1%29%29+-+cos%28x%29%28sqrt%28x%5E2%2Fsin%28x%29%5E2+%2B+C%5E2%29%29
\begin{align*}
    \frac{\partial}{\partial \theta} F(\theta, t) 
    &= \sqrt{(t\sin(\theta))^2 + \theta^2} + \frac{(1-\theta\cot(\theta))^2}{\sqrt{(t\sin(\theta))^2 + \theta^2}}. \\
    \intertext{
        First, we notice that $\frac{\partial}{\partial \theta} F \geq 0$, so $F$ is increasing.
        Second, we notice that, for $\theta \in [0, \pi/2]$,
    }
    \frac{\partial}{\partial \theta} F(\theta, t) 
    &\geq \theta\sqrt{(t\sin(\theta)/\theta)^2 + 1}
    \geq \theta\sqrt{t^2(4/\pi^2) + 1}.
\end{align*}
Integrating, we get that, for $\theta \in [0, \frac \pi 2]$, $F(\theta, t) - F(0, t) \geq \thalf\theta^2\sqrt{1+4t/\pi^2}$, and using that $F(0, t) = \log(t^{-1} + \sqrt{t^{-2} + 1}) - \sqrt{1 + t^2}$, we have the desired
\begin{align*}
    I_{n}(nt)
    &= \frac{1}{\pi} \int_0^\pi \exp(-n F(\theta, t))\dd\theta \\
    &\leq \frac{2}{\pi} \int_0^{\pi/2} \exp(-n F(\theta, t))\dd\theta \\
    &\leq \frac{2}{\pi} \int_0^{\pi/2} \exp(-\tfrac n 2\theta^2\sqrt{1 + 4t^2/\pi^2})\exp(-n F(0, t))\dd\theta \\
    &= \frac{\exp(-nF(0,t))}{\pi} \int_0^\pi \exp(-\thalf\theta^2\sqrt{n^2 + 4n^2t^2/\pi^2})\dd\theta \\
    &< \frac{\exp(-nF(0,t))}{\pi} \int_0^\infty \exp(-\thalf\theta^2\sqrt{n^2 + 4n^2t^2/\pi^2})\dd\theta \\
    &= \frac{\exp(-nF(0,t))}{\sqrt{2\pi\sqrt{n^2 + 4n^2t^2/\pi^2}}} \\
    &= \frac{\exp(\sqrt{n^2 + n^2t^2})}{(4\pi^2 n^2 + 16n^2 t^2)^{1/4}(t^{-1} + \sqrt{t^{-2} + 1})^n} \\
    &\leq \frac{\exp(\sqrt{n^2 + n^2t^2})(\sqrt{t^{-2} + 1} - t^{-1})^n}{2(n^2 + n^2 t^2)^{1/4}}. \tag*{\qedhere}
\end{align*}
\end{proof}
\section{Deferred proofs from Section~\ref{ssec:bounded_proof}}\label{app:more_csp}

\restatetinyellipse*
\begin{proof}
Recall from Theorem~\ref{thm:trefethen} that we can parameterize the boundary of $E_\rho$ as $\half(\rho + \rho^{-1}) \cos \theta + \half(\rho - \rho^{-1})\sin\theta$, for $\theta \in [0, 2\pi]$. Hence, to prove the stated inclusion it suffices to prove that $\half(\sigma + \sigma^{-1}) \ge (1+\delta)\half(\rho + \rho^{-1})$ and $\half(\sigma - \sigma^{-1}) \ge (1 + \delta)\half(\rho - \rho^{-1})$. By \cref{lem:bern-bounds},
    \begin{align*}
        \half(\sigma + \sigma^{-1}) - (1+\delta)\half(\rho + \rho^{-1})
        &= 1 + \frac{9(\alpha^2 + 2\alpha\sqrt{\delta} + \delta)}{2(1 + 3\alpha + 3\sqrt{\delta})} - (1+\delta)\Big(1 + \frac{\alpha^2}{2(1+\alpha)}\Big) \\
        &\geq 1 + \frac{2(\alpha^2 + 2\alpha\sqrt{\delta} + \delta)}{2(1 + \alpha)} - (1+\delta)\Big(1 + \frac{\alpha^2}{2(1+\alpha)}\Big) \\
        &= \frac{2(\alpha^2 + 2\alpha\sqrt{\delta} + \delta) - (1+\delta)\alpha^2 - 2\delta(1+\alpha)}{2(1 + \alpha)} \geq 0.
    \end{align*}
    Further, since $\sigma \geq (1+\delta)(1+\alpha)$ and $\sigma - \sigma^{-1}$ increases in $\sigma$,
    \begin{align*}
        \half(\sigma - \sigma^{-1}) - (1+\delta)\half(\rho - \rho^{-1})
        \geq \half((1+\delta)\rho - ((1+\delta)\rho)^{-1}) - (1+\delta)\half(\rho - \rho^{-1})
        \geq 0.
    \end{align*}
\end{proof}

\restatechebzbound*
\begin{proof}
The only points not lying on the boundary of a Bernstein ellipse are the interval $[-1, 1]$, where the Chebyshev polynomials are at most $1$. Otherwise, suppose $z \defeq x + \ii y$ lies on the boundary of the Bernstein ellipse $E_\rho$, for some $\rho > 1$. This implies that for some $\theta \in [-\pi, \pi]$,
\[\half (\rho + \rho^{-1}) \cos \theta = x,\; \half (\rho - \rho^{-1})\sin \theta = y,\]
which follows from the parameterization of the Bernstein ellipse in Theorem~\ref{thm:trefethen}. Let $s = \cos \theta$ and $t = \half(\rho + \rho^{-1})$. Noting that $\sqrt{t^2 - 1} = \half(\rho - \rho^{-1})$, we then have the system of equations
\[ts = x,\; \sqrt{1 - s^2}\sqrt{t^2 - 1} = y \implies t^4 - (1 + x^2 + y^2)t^2 + x^2 = 0.\]
Hence, solving a quadratic equation in $t^2 = \frac 1 4(\rho^2 + \rho^{-2}) + \half$ yields
\[\half(\rho^2 + \rho^{-2}) = x^2 + y^2 + \sqrt{(1 + x^2 + y^2)^2 - 4x^2} =: D, \]
and then $\rho = (D + \sqrt{D^2 - 1})^{\half}$. By definition, the Chebyshev polynomial satisfies 
\begin{equation}\label{eq:rhobound_suffices}T_n(z) \le \half (\rho^n + \rho^{-n}) \le \rho^n,\end{equation}
so it suffices to establish bounds on $\rho$.
Next, assume without loss of generality that $y \ge 0$, as Chebyshev polynomials are either odd or even and the stated conclusions are unsigned. We bound
\begin{align*}
D &= x^2 + y^2 + \sqrt{1 + x^4 + y^4 - 2x^2 + 2y^2 + 2x^2y^2} \\
&= x^2 + y^2 + \sqrt{(1 - x^2)^2 + 2y^2(1 + x^2) + y^4} \le x^2 + |1 - x^2| + O(y\sqrt{1 + x^2}).
\end{align*}
When $|x| \le 1$, then $D = 1 + O(y)$ and $\rho = 1 + O(\sqrt{y})$, establishing the conclusion via \eqref{eq:rhobound_suffices}. Otherwise, when $|x| > 1$, we have $D = 2x^2 - 1 + O(|x|y)$, and then
\begin{align*}
\rho &= \sqrt{2x^2 - 1 + \sqrt{4x^4 - 4x^2 + O(|x^3| y)}} \\
&= \sqrt{\Par{|x| + \sqrt{x^2 - 1}}^2 + O\Par{\sqrt{|x^3|y}}} = |x| + \sqrt{x^2 - 1} + O(\sqrt{|xy|}),
\end{align*}
again proving the desired claim via \eqref{eq:rhobound_suffices}.
\end{proof}

\restaterectbounds*
\begin{proof}
To see the first claim, let $z \in \R$. As previously discussed we have $\erf(s(\mu + z)), \erf(s(\mu - z)) \le 1$, giving the upper bound. For the lower bound, since $\erf$ is odd and increasing, $z \ge 0$ implies
\[\erf\Par{s\Par{\mu + z}} + \erf\Par{s\Par{\mu - z}} = \erf\Par{s\Par{\mu + z}} - \erf\Par{s\Par{-\mu + z}} \ge 0,\]
and a similar argument handles the case $z \le 0$. Next, for the second claim we first observe
\begin{equation}\label{eq:erf_close}
\begin{aligned}
|\erf(z) - \erf(x)| &= \frac{2}{\sqrt{\pi}} \Abs{\int_{t = x}^{t = x + \ii y} e^{-t^2} \dd t} \\
&= \frac{2e^{-x^2}}{\sqrt{\pi}} \Abs{\int_0^y e^{-2\ii xt + t^2} \dd t} \\
&\le \frac{2e^{-x^2}}{\sqrt{\pi}}\int_0^{|y|} e^{t^2} \dd t = e^{-x^2} |\erf(\ii y)|.
\end{aligned}
\end{equation}
The last line used the triangle inequality. Finally,
\begin{align*}
\left|\erf(z) - \erf(x)\right|&\le \half \Abs{\erf(s(\mu + z)) - \erf(s(\mu + x))} + \half \Abs{\erf(s(\mu - z)) - \erf(s(\mu - x))} \\
&\le \half\Par{e^{-s^2(\mu + x)^2} |\erf(\ii sy)| + e^{-s^2(\mu - x)^2} |\erf(-\ii sy)|},
\end{align*}
where we used \eqref{eq:erf_close}, and we conclude by noting $\mu + x, \mu - x \ge \mu - |x|$.
\end{proof}

\end{document}